\documentclass[12pt]{article}
\usepackage{epsfig,latexsym,amsfonts,amsmath,amsthm,amssymb,amsbsy,multirow,slashed,wasysym,textcomp,dsfont,comment,mathtools,cancel,cite,diagbox,datetime,appendix,BOONDOX-calo}
\usepackage{tocloft}
\usepackage{bbold}
\usepackage{sidecap}
\usepackage{adjustbox}
\usepackage{pgfplots}
\usepackage{pgfkeys}
\usepackage{oplotsymbl}
\usepackage{tikzpagenodes}
\usepackage{adforn}
\tikzstyle{bag} = [align=center]
\usetikzlibrary{shapes,backgrounds,arrows.meta,patterns,decorations.markings,bending,positioning, decorations.pathmorphing,cd} 
\tikzset{snake it/.style={decorate, decoration=snake}}
\usepackage{adjustbox}
\usepackage[normalem]{ulem}
\usepackage[mathscr]{eucal}
\usepackage[version=3]{mhchem}
\usepackage{physics}

\usepackage[hidelinks]{hyperref}
\usepackage{subcaption}
 \newcommand{\badat}{\begin{alignedat}}
 \newcommand{\eadat}{\end{alignedat}}
 
 \def\be{\begin{equation}}
\def\ee{\end{equation}}

\usepackage{color}

\newcommand{\pink}[1]{\textcolor{\pink}{#1}}

\definecolor{dblue}{rgb}{0.2,0.50,0.80}

\usepackage[utf8]{inputenc}

\tikzset{snake it/.style={decorate, decoration=snake}}

\usepackage{tensor}
\usepackage{amsmath}
\usepackage{amssymb}
\usepackage{mathrsfs}
\usepackage[normalem]{ulem}
\usepackage{bbm} 
\usepackage{graphicx}
\usepackage{stmaryrd}
\usepackage{xcolor}
\usepackage{mathtools}
\usepackage{xfrac}
\theoremstyle{definition}
\newtheorem{defn}{Definition}
\newtheorem{prop}{Proposition}
\newtheorem{thm}{Theorem}

\numberwithin{equation}{section} 
\pgfplotsset{compat=1.17} 
\begin{document}

\begin{titlepage}

\unitlength = 1mm
\ \\
\vskip 2cm
\begin{center}
{\LARGE{\textsc{Two Roads to Fortuity in ABJM Theory}}}

\vspace{0.8cm}
Connor Behan$^{1,2}$ and Leonardo Pipolo de Gioia$^2$
\vspace{0.7cm}\\
\small{${}^{1}$\textit{Perimeter Institute for Theoretical Physics, 31 Caroline St. N, N2L 2Y5,  Waterloo, Canada}
\vspace{5pt}\\
${}^2$\textit{ICTP South American Institute for Fundamental Research\\IFT-UNESP, S\~{a}o Paulo, SP Brazil 01440-070}} 
\vspace{20pt}

\begin{abstract}
A recently proposed addition to the holographic dictionary connects extremal black holes to fortuitous operators --- those which are only supersymmetric for sufficiently small values of the central charge. The most efficient techniques for finding them come from studying the cohomology of a nilpotent supercharge. We explore two aspects of this problem in weakly-coupled ABJM theory, where the gauge group is $\mathrm{U}(N) \times \mathrm{U}(N)$ and the Chern-Simons level is taken to be large. Adapting an algorithm which has been used to great effect in $\mathcal{N} = 4$ Super Yang-Mills, we enumerate 244 low-lying fortuitous operators and sort them into multiplets of the centralizer algebra. This leads to the construction of two leading fortuitous representatives for $N = 3$ which are subleading for $N = 2$. In the second part of this work, we identify a truncation of ABJM theory where the action of the one-loop supercharge matches the one in the BMN subsector of $\mathcal{N} = 4$ Super Yang-Mills. This allows a known infinite tower of representatives to be lifted from one theory to the other.
\end{abstract}
\vspace{0.5cm}
\end{center}
\vspace{0.8cm}

\end{titlepage}
\tableofcontents

\section{Introduction}

A lasting impact of string theory is its ability to account for the microstates of black holes \cite{bmmw23}. The first demonstration of this came when Strominger and Vafa reproduced the Bekenstein-Hawking entropy formula for a family of supersymmetric black holes in a compactification of type IIA string theory down to five dimensions \cite{hep-th/9601029}. This was done by going to weak coupling and performing a T-duality transformation so that the black hole microstates become part of the spectrum of a two-dimensional conformal field theory (CFT). Specifically, this was the worldvolume theory on a stack of D1-branes which were part of a bound state also involving D3's and D5's \cite{hep-th/9601029}. States were counted in this CFT by estimating the elliptic genus
\begin{align}
I(q, y) = \text{Tr}_{\text{RR}} \left [ (-1)^F q^{L_0 - c/24} \bar{q}^{\bar{L}_0 - c/24} y^{J^{(z)}_0} \right ]. \label{elliptic-genus}
\end{align}
%where $q = e^{2\pi i \tau}$ and $y = e^{2\pi i z}$.
Since the trace is in the Ramond-Ramond sector, the supercurrent has a zero mode $\bar{G}_0$ which commutes with $\bar{L}_0$ but changes the fermion number $F$. As a result, only dilation eigenstates annihilated by $\bar{L}_0$ (known as $\frac{1}{2}$-BPS states in this context) can survive. On the other hand, R-symmetry acts non-trivially on the supercurrent so the presence of $J^{(z)}_0$ allows $L_0$ to take any value. The crucial property which makes \eqref{elliptic-genus} an index is that there is also cancellation between $\frac{1}{2}$-BPS states which have the right quantum numbers to combine into long multiplets and lift above the BPS bound. One can therefore be sure that \eqref{elliptic-genus} stays protected over the full moduli space of a CFT.

The AdS/CFT correspondence \cite{hep-th/9711200} has allowed checks of this sort to be done in several other cases. Interestingly, the counting of black hole microstates in the most famous version of the correspondence --- the duality between $\mathcal{N} = 4$ Super Yang-Mills (SYM) and type IIB string theory in $AdS_5 \times S^5$ --- followed a somewhat circuitous route. Although $\frac{1}{16}$-BPS (annihilated by one Poincar\'{e} supercharge) black hole solutions of 5d gauged supergravity were first found in \cite{hep-th/0401042,hep-th/0401129}, an early attempt to recover their entropy from a superconformal index calculation in \cite{kmmr07} failed due to large cancellations between bosonic and fermionic states. What finally led to a match was the realization that the index should be evaluated in the neighbourhood of complex saddle points \cite{1705.05383,1810.11442,1810.12067,1812.09613}. One of the works which inspired this renewed look at $\mathcal{N} = 4$ SYM was a study of a different (and less supersymmetric) holographic CFT which can be described as a topologically twisted version of ABJM theory \cite{0806.1218}. As shown in \cite{1511.04085,1608.07294}, its superconformal index captures the entropy of black holes which live in magnetic $AdS_4$ and are annihilated by two supercharges. As far as we are aware, a similar study of ordinary ABJM theory has not been done because the analogous supergravity solutions for $AdS_4 \times \mathbb{CP}^3$ are not known. This paper will approach the study of black holes described by ABJM theory from a different perspective.

To do so, we will
%use the AdS/CFT correspondence to develop bulk intuition about
focus on
explicit states constructed using the field theory at weak coupling. To have some hope of also learning about strong coupling, where stringy effects subside, one can further restrict to states preserving at least some supersymmetry (SUSY). Specifically, consider $\left | \psi \right >$ which is annihilated by a nilpotent supercharge $Q$ and its BPZ-conjugate $Q^\dagger$. In addition to being $Q$-closed by definition, it is also straightforward to argue that $\left | \psi \right >$ should not be $Q$-exact. If it were, we would be able to write
\begin{align}
\left | \psi \right > = Q \left | \psi' \right > \Longrightarrow Q^\dagger Q \left | \psi' \right > = 0
\end{align}
and find that $\left | \psi \right >$ would need to have a vanishing norm. This suggests that what we are really looking at is the \textit{cohomology} of $Q$. Enumerating this cohomology can reveal geometric information as shown by \cite{0803.4183} for $\mathcal{N} = 4$ SYM with gauge groups of low rank. In particular, one can look for cohomology classes which are not dual to a gas of gravitons in the bulk. A computerized search in \cite{1305.6314} failed to find examples for ${\rm SU}(2)$ and ${\rm SU}(3)$, leading to a possible reason to doubt their existence. After many years however, this conclusion was overturned in \cite{2209.06728} which used a more efficient enumeration to compute the full and multi-graviton partition functions at high enough levels for a discrepancy to appear. The associated operator is annihilated by $Q$ only as a result of trace relations \cite{2209.12696} and has therefore been termed fortuitous. The physical picture corresponding to this is that gravity must be sufficiently strong for all black holes, including BPS ones, to form. Sending Newton's constant to zero should therefore cause the dual operators to lift above the BPS bound \cite{bmv23}.

More generally, it has now been proposed that all cohomology classes in a holographic CFT should be thought of as either monotone or fortuitous --- precise definitions are given in \cite{cl24}, which we will soon review. The monotone classes include coherent states of gravitons while the fortuitous classes include black holes, but it is sometimes not obvious where more exotic geometries belong. Alongside $\mathcal{N} = 4$ SYM, the fortuity mechanism has also been investigated in matrix models \cite{2504.14181,2511.00790}, supersymmetric SYK models \cite{2412.06902}, the D1-D5 CFT \cite{2501.05448,2505.14888,2511.23294} and the Leigh-Strassler CFT \cite{ck25}. For the present paper, we are most interested in 3d gauge theories with a known gravity dual. These include ABJ theory \cite{0807.4924} dual to higher-spin gravity in $AdS_4$ and ABJM theory \cite{0806.1218} dual to type IIA string theory in $AdS_4 \times \mathbb{CP}^3$. For the case of ABJ theory, the authors of \cite{kllo25} constructed some fortuitous states and found that graviton composites were joined by higher-spin composites within the monotone sector. Most recently, the first fortuitous state in weakly-coupled ABJM theory was found in \cite{bsvz25} which has overlap with this work. Our goal is to uncover additional structure in the fortuitous Hilbert space by combining systematic searches at low level with more sporadic searches at high level. In particular, we will exploit a previously unnoticed relation between the enumeration problems in ABJM and $\mathcal{N} = 4$ SYM.

This paper is organized as follows. Section \ref{sec:holographic-coverings} starts with the minimal amount of structure needed to define a holographic CFT possessing a nilpotent supercharge. After reviewing the formalism of \cite{cl24}, we consider trace relations for different types of matrices which have appeared in the literature on fortuitous operators. Section \ref{sec:abjm-setup} then discusses the field content of ABJM theory and studies the privileged sector of fields on which BPS operators are allowed to depend. We point out connections to ${\cal N}=4$ SYM where appropriate, especially regarding recent work on coupling dependence of the supercharge cohomology \cite{cl25a,cl25b,bk25} which has yet to be done for ABJM. Section \ref{sec:abjm-details} contains both of our main results. The first is a level-by-level counting procedure, implemented in a public code, which has allowed us to explicitly construct three fortuitous cohomology classes and identify the charge sectors associated with several more. The second is a map from SYM BPS operators to ABJM BPS operators which we combine with the results of \cite{2304.10155} to construct
-%fortuitous cohomology classes in an infinite family.
-an infinite family of cohomology classes which we conjecture to be fortuitous.
-Section \ref{sec:conclusion} concludes with a summary of these results and potential future directions.

\section{Holographic Coverings and Black Hole States}\label{sec:holographic-coverings}

In this section, we are going to review the classification of the BPS operators proposed in \cite{cl24}, which separates them into two classes: the \textit{monotone operators}, which are presumably dual to smooth horizonless geometries, and the \textit{fortuitous operators}, which are
%presumably dual to geometries containing
candidates for
black holes. The definitions are cohomological in nature, and are phrased in the language of cochain complexes on abelian categories.

Let ${\cal H}_N$ be the Hilbert space of an SCFT labelled by $N$ in an infinite sequence. Suppose that $Q_N$ is a nilpotent supercharge acting on ${\cal H}_N$ with BPZ adjoint $S_N = Q_N^\dagger$. We further denote $\Delta \equiv 2\{Q_N,S_N\}$.
\begin{defn}
    A state $\psi\in {\cal H}_N$ is called a \textbf{BPS state} if $\Delta \psi =0$ and the subspace of BPS states is denoted ${\cal H}_N^{\rm BPS} \equiv \ker \Delta$.
\end{defn}
In what follows we will provide a detailed account of the classification of BPS states into monotone and fortuitous. This classification starts with the definition of a holographic covering of a sequence of Hilbert spaces.
\begin{defn}\label{defn:covering}
    Given a sequence of vector spaces ${\cal H}_N$ for $N\in \mathbb{N}$, a covering consists of a covering vector space $\tilde{\cal H}$ and a sequence of subspaces $I_N\subset \tilde{\cal H}$ such that ${\cal H}_N\simeq \tilde{\cal H}/I_N$. Equivalently, there is a short exact sequence
    \begin{equation}
        \begin{tikzcd}
            0\arrow[r] & I_N\arrow[r, "\iota_N"] & \tilde{\cal H}\arrow[r, "\pi"] & {\cal H}_N\arrow[r]& 0.
        \end{tikzcd}
    \end{equation}
\end{defn}
\begin{defn}\label{defn:strict-covering}
    Let ${\cal H}_N$ be a sequence of vector spaces for $N\in \mathbb{N}$ and $\left(\tilde{\cal H},(I_N)_{N\in \mathbb{N}}\right)$ a covering. We say that the covering is \textbf{strict} if and only if $I_{N+1}\subset I_N$ and for any $v\in \tilde{\cal H}$ there is $N\in \mathbb{N}$ such that $v\notin I_N$.
\end{defn}
\begin{defn}
    Let ${\cal H}_N$ be a sequence of vector  spaces for $N\in \mathbb{N}$ and $\left(\tilde{\cal H}, (I_N)_{N\in \mathbb{N}}\right)$ a covering. We say that a sequence of operators $T_N \in \operatorname{End}({\cal H}_N)$ lifts to the covering if there exists $\tilde{T}\in \operatorname{End}(\tilde{\cal H})$ such that the following diagram commutes.
    \begin{equation}
        \begin{tikzcd}
            \tilde{\cal H}\arrow[r, "\exists \ \tilde{T}",dashed]\arrow[d,"\pi_N"']&\tilde{\cal H}\arrow[d,"\pi_N"]\\
            {\cal H}_N\arrow[r,"T_N"']&{\cal H}_N
        \end{tikzcd}
    \end{equation}
\end{defn}
\begin{defn}\label{defn:cover-symmetry}
    Let ${\cal H}_N$ be a sequence of vector spaces for $N\in \mathbb{N}$ and $\left(\tilde{\cal H},(I_N)_{N\in \mathbb{N}}\right)$ a covering. Let $\mathfrak{g}$ be a Lie superalgebra and $\rho_N : \mathfrak{g}\to \operatorname{End}({\cal H}_N)$ representations on each ${\cal H}_N$. We say that the sequence of representations lifts to the covering if there exists a Lie superalgebra representation $\tilde{\rho} : \mathfrak{g}\to \operatorname{End}(\tilde{\cal H})$ such that $\tilde{\rho}(X)$ is a lift of the sequence $\rho_N(X)$ for all $X\in \mathfrak{g}$.
\end{defn}

Notice that definitions 1 and 2 are exactly the same as in \cite{cl24}, but our definitions 3 and 4 are meant to adapt and generalize definition 3 in that reference. In particular we consider Lie superalgebras instead of compact Lie groups. Moreover, definition 3 in \cite{cl24} had a superfluous condition that $I_N$ be invariant under the lifted representation. We can easily see, however, that this is implied by commutativity of the above diagram. Definition 1 implies that $I_N \simeq \ker \pi_N$ and as such $v\in I_N$ if and only if $\pi_N(v)=0$. The diagram, on the other hand, demands $\pi_N\circ \tilde{\rho}(X) = \rho_N(X)\circ \pi_N$. Therefore if $v\in I_N$ the right-hand side vanishes and so does the left-hand side, showing that $\tilde{\rho}(X)v \in \ker \pi_N$. As such, $v\in I_N$ implies $\tilde{\rho}(X)v\in I_N$ for all $X\in \mathfrak{g}$. Note that the same argument applies \textit{mutatis mutandis} to group representations.

These definitions are meant to formalize and generalize some standard techniques for supersymmetric gauge theories. These often have fundamental fields which are matrix-valued and primary operators which can be constructed as multi-traces of products of these fields at weak coupling. The strict inclusion $I_N \subset I_{N - 1}$ reflects the fact that trace relations for square matrices of size $N$ also hold for those of size $N - 1$ but not \textit{vice versa}.
%In these theories, we have single trace operators $\operatorname{Tr}[\Phi_1\cdots \Phi_L]$, where each of the $\Phi_i$ is taken from the set of fundamental fields. From the single trace operators we can now take arbitrary products and arbitrary linear combinations. Importantly, if the gauge group is ${\rm SU}(N)$, the $\Phi_i$ are $N\times N$ matrices, but we can abstract that away with the free algebra construction.
Just as we assume $Q_N \in \operatorname{End}({\cal H}_N)$ (and in fact all superconformal generators) exist for all $N$, we will also assume that they lift to the covering space.
%Having this intuition in mind, it should now be clear that the superconformal symmetry algebra lifts to the covering: the action of the generators of the superconformal algebra are independent of the size of the matrices.

\subsection{BPS operators and supercharge cohomology}\label{sec:bps-operators}

Under the assumption that ${\cal H}_N$ admits a covering in the sense of definition \ref{defn:covering} we are going to further assume that there exists a lift of the supercharge, in other words, there exists $\tilde{Q} : \tilde{\cal H}\to \tilde{\cal H}$ such that the following diagram commutes.
\begin{equation}
    \begin{tikzcd}
        \tilde{\cal H}\arrow[d,"\pi_N"']\arrow[r, "\tilde{Q}"] & \tilde{\cal H}\arrow[d,"\pi_N"]\\
        {\cal H}_N \arrow[r,"Q_N"']& {\cal H}_N
    \end{tikzcd}
\end{equation}
The intuition behind this assumption is very clear: the action of a supercharge on the fundamental fields of an SCFT is independent of $N$. In other words: the supercharge ``does not see'' that the fundamental fields are matrices of size $N\times N$. As such it naturally lifts to the space of formal multi-traces.

Now, choose $E_N\in \operatorname{End}({\cal H}_N)$ such that $[E_N,Q_N]=Q_N$ and $E_N^\dagger = E_N$.\footnote{Typically, $E_N$ is a linear combination of Cartans from the maximal bosonic subalgebra of the superconformal algebra.}
The $E_N$ must also have eigenvalues of the form $E_n = E_0 + n$ for $n\in \mathbb{Z}$ or $n\in \mathbb{Z}_{\geq}$. Since $E_N$ is hermitian, ${\cal H}_N$ decomposes into eigenspaces of $E_N$ and if ${\cal H}_N^{(n)}$ is the eigenspace of eigenvalue $n\in \mathbb{Z}$ we get a cochain complex
\begin{equation}
    \begin{tikzcd}
        \cdots \arrow[r,"Q_N"] & {\cal H}^{(n-1)}_N\arrow[r,"Q_N"] & {\cal H}^{(n)}_N \arrow[r,"Q_N"] & {\cal H}^{(n+1)}_N\arrow[r,"Q_N"] & \cdots
    \end{tikzcd}.
\end{equation}
Let $H^n_{Q_N}({\cal H}_N)$ be the $n$-th cohomology space of that complex:
\begin{equation}
    H^n_{Q_N}({\cal H}_N) \equiv \dfrac{\ker \left(Q_N : {\cal H}^{(n)}_N\to {\cal H}^{(n+1)}_N\right)}{\operatorname{im} \left(Q_N : {\cal H}^{(n-1)}_N\to {\cal H}^{(n)}_N\right)}.
\end{equation}
From $H^n_{Q_N}({\cal H}_N)$ we obtain a new graded vector space
\begin{equation}
    H^\ast_{Q_N}({\cal H}_N) = \bigoplus_{n\in \mathbb{Z}}H^n_{Q_N}({\cal H}_N)
\end{equation}
and this space is isomorphic to ${\cal H}_N^{\rm BPS}\subset {\cal H}_N$, the subspace of BPS operators defined to be the ones annihilated by $\Delta \equiv 2\{Q,Q^\dagger\}$. We review the standard argument for this isomorphism in Appendix \ref{app:bps-cohomology-isomorphism}. Elements of the cohomology $H^\ast_{Q_N}({\cal H}_N)$ may be distinguished by how the supercharge acts on them when we lift to the covering space.

\begin{defn}
    Let ${\cal H}_N$ be a sequence of vector spaces for $N\in \mathbb{N}$ and $\left(\tilde{\cal H}, (I_N)_{N\in \mathbb{N}}\right)$ a covering. Let $Q_N\in \operatorname{End}({\cal H}_N)$ be nilpotent admitting a lift $\tilde{Q}\in \operatorname{End}(\tilde{\cal H})$ to the covering. Given $\psi\in {\cal H}_N$ with lift $\tilde{\psi}\in \tilde{\cal H}$, we say that $\psi$ is \textbf{monotone} if $\tilde{Q}\tilde{\psi} =0$ and we say that $\psi$ is \textbf{fortuitous} if $\tilde{Q}\tilde{\psi}\neq 0$ and $\tilde{Q}\tilde{\psi} \in I_N$.
\end{defn}

This is the main definition we are going to employ: monotone states should be multi-graviton states, whereas the default interpretation of a fortuitous state should be a black hole. The quest for black hole states then takes a very precise form: at a given rank $N$, find an operator ${\cal O}$ on ${\cal H}_N$ with a lift $\tilde{\cal O}$ to $\tilde{\cal H}$ such that $\tilde{Q}\tilde{\cal O}\neq 0$ and $\tilde{Q}\tilde{\cal O}\in I_N$.

\subsection{Specializing to traces}

In the literature on fortuity, the formalism above has been specialized in at least two ways. For the D1-D5 CFT studied in \cite{2501.05448,2505.14888,2511.23294}, one uses the $I_N$ to implement restrictions on cycle shapes associated with twisted sectors. In all other cases we are aware of, these ideals encode trace relations on matrices. This paper will be an example of the latter. More precisely, we will let $\tilde{\cal H}$ denote the vector space spanned by products of formal traces (obeying cyclicity) in a set of fundamental fields. We may write
\begin{align}
\mathcal{O} = \sum_{i = 1}^M c_i \prod_{j=1}^{\deg_{\rm MT}(i)} \text{Tr} \left ( \prod_{k=1}^{\deg_{\rm ST}(i,j)} X_{i,j,k} \right )
\end{align}
to express a general element of $\tilde{\cal H}$.\footnote{Explaining the notation, ${\cal O}$ is a linear combination of $M$ multi-traces. The $m$-th multi-trace has multi-trace degree $\deg_{\rm MT}(m)$, which is the number of single-traces being multiplied. The $n$-th single-trace in the $m$-th multi-trace has single-tracee degree $\deg_{\rm ST}(m,n)$, which is the length of the corresponding word.} Taking $I_N$ to be the ideal of trace relations for a particular value of $N$, $\tilde{\cal H}/I_N$ yields a quotient space on which the trace relations hold true. This space is then equivalent to the space ${\cal H}_N$ where the multi-traces are not formal, but rather constructed from \textit{concrete} $N\times N$ matrices.
To close this section, we would like to discuss two bits of extra structure which may arise in this picture. The first is the fact that one sequence ${\cal H}_N$ can sometimes contain information about a \textit{different} sequence ${\cal H}^\prime_N$ even when their coverings are not the same. The second is that some SCFTs require one to consider matrices which are not square.

The first issue occurs when the relations in $(I_N)_{N \in \mathbb{N}}$ and $(I'_N)_{N \in \mathbb{N}}$ apply to matrices which furnish the adjoint representations of Lie algebras in $(\mathfrak{g}_N)_{N \in \mathbb{N}}$ and $(\mathfrak{g}'_N)_{N \in \mathbb{N}}$. If $\mathfrak{g}_N$ is a subalgebra of $\mathfrak{g}'_N$, lifting an element of $H^\ast_{Q_N}({\cal H}_N)$ to $H^\ast_{Q'_N}({\cal H}^\prime_N)$ is well defined. Our main example will be $\mathfrak{u}(N)$ and $\mathfrak{su}(N)$, thought of as gauge algebras. These make it natural for us to define a map by
\begin{align}
f(X) = X - \frac{1}{N} \text{Tr}(X) \textbf{1}. \label{subtract-trace}
\end{align}
For any such map acting on letters, we extend it to act on $\tilde{\cal H}$ by
\begin{align}
F(\mathcal{O}) = \sum_{i = 1}^M c_i \prod_{j=1}^{\deg_{\rm MT}(i)} \text{Tr} \left ( \prod_{k=1}^{\deg_{\rm ST}(i,j)} f(X_{i,j,k}) \right ).
\end{align}
We argue that if the action of $Q$ on letters commutes with $f$, then the action of $Q$ on formal multi-traces commutes with $F$.
\begin{prop}
    Let $f:{\cal X}\to {\cal X}$ be a mapping of letters and let $F:\tilde{\cal H}\to \tilde{\cal H}$ be the associated induced mapping of formal multi-traces on these letters. If $Q\in \operatorname{End}(\tilde{\cal H})$ is a linear derivation induced by its corresponding action on letters, and if $Q\circ f = f\circ Q$ then $Q\circ F = F\circ Q$.
\end{prop}
\begin{proof}
    The proof is done in steps. First consider just single traces. In that case we have
    \begin{equation}
        Q \text{Tr}(a_1\cdots a_n) = \sum_{i=1}^n \text{Tr}(a_1\cdots Qa_i\cdots a_n).
    \end{equation}
    We now act with $F$ as
    \begin{equation}
        \begin{split}
	    F[Q\text{Tr}(a_1\cdots a_n)] &= \sum_{i=1}^n \text{Tr}(f(a_1)\cdots f(Qa_i)\cdots f(a_n)) \\
            &= \sum_{i=1}^n \text{Tr}(f(a_1)\cdots Qf(a_i)\cdots a_n)\\
            &= Q \text{Tr}(f(a_1)\cdots f(a_n))\\
            &= Q F[\text{Tr}(a_1\cdots a_n)],
        \end{split}
    \end{equation}
    where we have used the assumption that $f\circ Q = Q\circ f$ on letters.
    As a result, $F\circ Q = Q\circ F$ on single traces. Now consider a product of two multi-traces. The $Q$ action gives
    \begin{equation}
        Q[\text{Tr}(a_1\cdots a_n)\text{Tr}(b_1\cdots b_n)]=[Q\text{Tr}(a_1\cdots a_n)]\text{Tr}(b_1\cdots b_n)+\text{Tr}(a_1\cdots a_n)[Q\text{Tr}(b_1\cdots b_n)]
    \end{equation}
    and the $F$ action is
    \begin{equation}
        \begin{split}
            FQ[\text{Tr}(a_1\cdots a_n)\text{Tr}(b_1\cdots b_n)] &= [FQ\text{Tr}(a_1\cdots a_n)]F\text{Tr}(b_1\cdots b_n)+F\text{Tr}(a_1\cdots a_n)[FQ\text{Tr}(b_1\cdots b_n)] \\
            &=[QF\text{Tr}(a_1\cdots a_n)]F\text{Tr}(b_1\cdots b_n)+F\text{Tr}(a_1\cdots a_n)[QF\text{Tr}(b_1\cdots b_n)]\\
            &=Q [F\text{Tr}(a_1\cdots a_n)F\text{Tr}(b_1\cdots b_n)]\\
            &= Q F [\text{Tr}(a_1\cdots a_n)\text{Tr}(b_1\cdots b_n)].
        \end{split}
    \end{equation}
    where we have invoked commutativity at the level of the single traces.
    This shows that commutativity also holds for products of single traces. Since $F$ is linear and multiplicative, and since $\tilde{\cal H}$ is the free non-commutative algebra generated by single traces, it follows that $Q\circ F = F\circ Q$.
\end{proof}
We can now draw a useful conclusion which is specific to the map which projects onto the traceless part of a letter.
\begin{prop}
    Let $f$ act on letters by \eqref{subtract-trace}. Then $F(I_{\rm SU(N)})\subset I_{\rm U(N)}$ where $I_{\rm SU(N)}$ and $I_{\rm U(N)}$ are ideals of trace relations for the indicated groups.
\end{prop}
\begin{proof}
    The difference between ${\rm SU}(N)$ and ${\rm U}(N)$ trace relations is just that ${\rm SU}(N)$ has the additional zero trace condition. This means that at the level of ideals
    \begin{equation}
        I_{\rm SU(N)} = I_{\rm U(N)} + J,\quad J =\langle \text{Tr}(X) : X \in{\cal X} \rangle \label{ideal-sum}
    \end{equation}
    where again ${\cal X}$ is the set of letters. To prove the inclusion we can show that $F(I_{\rm U(N)})\subset I_{\rm U(N)}$ and that $F(J)\subset I_{\rm U(N)} $. The first inclusion is easy to understand: an element $X \in I_{\rm U(N)}$ is a formal multi-trace that vanishes when letters are replaced by arbitrary ${\rm U}(N)$ matrices. At this level, $F$ acts on each letter by subtracting the trace and therefore produces a ${\rm U}(N)$ matrix once again. As such, the expression $F(X)$ still vanishes for ${\rm U}(N)$ matrices and the ideal is preserved. As for $J$, it has the single generator $\text{Tr}(X)$, and we can see that
    \begin{equation}
        F(\text{Tr}(X)) = \text{Tr}(X) \left [ 1 - \frac{1}{N} \text{Tr}(\textbf{1}) \right ].
    \end{equation}
    The trace of the identity gives $N$ in the quotient by $I_{\rm U(N)}$. As such $F(J)\subset I_{\rm U(N)}$ and altogether we showed that $F(I_{\rm SU(N)})\subset I_{\rm U(N)}$.
\end{proof}
Definition 6 makes it clear that a fortuitous operator in a ${\rm U(N)}$ gauge theory is necessarily fortuitous in the corresponding ${\rm SU(N)}$ gauge theory as well. This is because the latter has a strictly larger set of trace relations, \textit{i.e.} $I_{\rm U(N)} \subset I_{\rm SU(N)}$. Proposition 2 is effectively a converse --- a given ${\rm SU}(N)$ fortuitous operator may be lifted to a ${\rm U(N)}$ one straightforwardly. An example for $\mathcal{N} = 4$ SYM with ${\rm SU}(2)$ gauge group is the black hole candidate first detected by the enumeration in \cite{2209.06728}. The expression for it found in \cite{2209.12696} is
\begin{align}
%\mathcal{O} = \text{Tr}(\phi^1 \psi_1 - \phi^2 \psi_2) \text{Tr}(\phi^1 \psi_3) \text{Tr}(\psi_2 \psi_1 \psi_1) + \text{cyclic}
\mathcal{O} = \epsilon^{m_1 m_2 m_3} \text{Tr}(\phi^{m_4} \psi_{m_1}) \text{Tr}(\phi^{m_5} \psi_{m_2}) \text{Tr}(\psi_{m_3} [\psi_{m_4}, \psi_{m_5}]) \label{sym-fort}
\end{align}
in terms of three bosonic and three fermionic BPS letters.
%whose labels are permuted together.
Using the supercharge action which will be given in the next section, we have verified that $\mathcal{O}$ and $F(\mathcal{O})$ are both fortuitous. The fact that ${\rm U}(N)$ versus ${\rm SU}(N)$ is innocuous in $\mathcal{N} = 4$ SYM fortuity will play an important role for us later on.

The second issue worth mentioning deals with matter transforming in the bifundamental of two gauge groups. This is important for the so called ABJ triality \cite{1207.4485,2103.01969} which defines different scaling limits for a Chern-Simons matter theory with gauge group ${\rm U}(N)_k \times {\rm U}(M + N)_{-k}$. Our main interest is ABJM theory with $M = 0$ but it is also possible to fix $N$ while scaling $M$. This leads to a covering space $\tilde{\cal H}$ where the formal multi-traces obey more relations than just cyclicity and the ideal denoted by $J$ in \eqref{ideal-sum}. This exotic covering is important for interpreting the results of \cite{kllo25} which considered the ABJ theory with ${\rm U}(N)_k \times {\rm U}(1)_{-k}$.

As an example, we can consider the lowest-lying fortuitous operators in $N = 2$ ABJM theory. These were found in \cite{bsvz25} and will also be constructed later in this paper. In our preferred notation, these four-fold degenerate operators may be expressed as
\begin{align}
\mathcal{O}^a_b = \epsilon^{cd} \text{Tr}(\bar{\psi}_c \phi_d \bar{\phi}^a \phi_b) - \epsilon_{cd} \text{Tr}(\psi^c \bar{\phi}^d \phi_b \bar{\phi}^a) - \alpha \text{Tr}(\bar{\phi}^a \phi_b) \left [ \epsilon^{cd} \text{Tr}(\bar{\psi}_c \phi_d) - \epsilon_{cd} \text{Tr}(\psi^c \bar{\phi}^d) \right ] \label{abjm-fort}
\end{align}
for a properly chosen coefficient $\alpha$. It is easy to read off quantum numbers and find that this operator has $E + J = 3$.\footnote{As we will see, the nilpotent supercharge preserves the sum of energy and angular momentum eigenvalues, denoted $E$ and $J$ respectively.} Naively, this conflicts with the statement in \cite{kllo25} that the threshold fortuitous operator in $N = 2$ ABJ theory has $E + J = 8$. After all, the trace relations in ${\rm U}(2)_k \times {\rm U}(2)_{-k}$ are a subset of those in ${\rm U}(2)_k \times {\rm U}(1)_{-k}$. The resolution comes from the fact that the extra relations in question survive the $N \to \infty$ limit and thus change the definition of fortuity. The two possibilities are as follows.
\begin{enumerate}
\item If the matter fields are $N \times N$ matrices for unspecified $N$, one can act on \eqref{abjm-fort} with $\tilde{Q}$ to find an expression in terms of formal traces. This expression never belongs to $I_{N\times N}$ for $N \geq 3$ but it belongs to $I_{2\times 2}$ when $\alpha = \frac{3}{4}$. As a result it also belongs to $I_{N\times 1}$ for all $N\geq 2$.
\item If the matter fields are $N \times 1$ instead, the $\tilde{Q}$ action on \eqref{abjm-fort} belongs to $I_{N\times 1}$ for all $N$ when $\alpha = 1$.
\end{enumerate}
This constrained covering also explains the observation in \cite{kllo25} that ABJ theory's gravitons (defined as $Q$-closed single traces) do not behave as an ideal gas. When the supercharge action preserves the number of traces, the theory of cyclic cohomology \cite{2306.01039} guarantees that composites of graviton operators span the monotone Hilbert space. This generalizes earlier work by \cite{0712.2714,1305.6314} which found that the BPS and multi-graviton Hilbert spaces agreed at infinite $N$ in $\mathcal{N} = 4$ SYM. In ABJ theory, the trace preserving assumption fails and non-graviton degrees of freedom need to be considered as well.

\section{BPS States at Weak Coupling}\label{sec:abjm-setup}

Having introduced the holographic covering formalism, it is time to move onto a concrete application --- identifying fortuitous operators in 3d $\mathcal{N} = 6$ ABJM theory. To do so, we will construct cohomology classes out of general traces and $Q$-closed traces in order to compare dimensionalities. This section will explain the setup which goes into this. Structural similarities between this problem and the now standard one in $\mathcal{N} = 4$ SYM will also be reviewed in order to open up another line of attack. Even though the work of \cite{2209.06728} shows that the leading fortuitous operator \eqref{abjm-fort} in ${\rm U}(2)_k \times {\rm U}(2)_{-k}$ ABJM does not have an SYM counterpart, the converse need not be true. It is plausible that one can find an ABJM counterpart to the leading fortuitous operator \eqref{sym-fort} in ${\rm U}(2)$ SYM.

\subsection{Basics of ABJM theory}

ABJM theory is a Chern-Simons matter theory with gauge group ${\rm U}(N)_k \times {\rm U}(N)_{-k}$ where $k$ is the Chern-Simons level. It was first proposed in \cite{0806.1218} as a holographic description of a stack of $N$ M2 branes. The theory exhibits ${\cal N}=6$ superconformal symmetry which is non-trivially enhanced to ${\cal N}=8$ when $k=1,2$. The theory is conjectured to be dual to M theory on $AdS_4\times S^7/\mathbb{Z}_k$, while in the 't Hooft limit $N,k\to \infty$ with $\lambda = \frac{N}{k}$ fixed, the theory is dual to type IIA superstring theory on $AdS_4\times \mathbb{CP}^3$.

The matter content of ABJM theory consists of bosons $\phi_A, \bar{\phi}^A$ and fermions $\psi^A, \bar{\psi}_A$. With respect to gauge symmetry, a bar denotes a bifundamental in $(\Box, \bar{\Box})$ and no bar denotes a bifundamental in $(\bar{\Box}, \Box)$. With respect to R-symmetry, a lower $A = 1, \dots, 4$ is for spinors and an upper $A = 1, \dots, 4$ is for conjugate spinors. We also have the gauge fields $A\bar A_\mu$ and $\bar A A_\mu$. These are often just called $A_\mu$ and $\bar{A}_\mu$ respectively but we prefer the verbose notation since it is easy to remember that barred and unbarred quantities must alternate inside a trace. As can be checked in \cite{0806.4977}, the action describing how these ingredients are coupled together is
\begin{align}
S = \frac{k}{2\pi} \int d^3x \, \text{Tr} (L_{\text{kin}} + L_{\text{CS}} + L_{\text{pot}})
\end{align}
where
\begin{align}
L_{\text{kin}} &= D_\mu \bar{\phi}^A D^\mu \phi_A - i \bar{\psi}_A \gamma^\mu D_\mu \psi^A \\
L_{\text{CS}} &= \epsilon^{\mu\nu\rho} \left [ \frac{1}{2} A\bar{A}_\mu \partial_\nu A\bar{A}_\rho + \frac{1}{3} A\bar{A}_\mu A\bar{A}_\nu A\bar{A}_\rho - (A\bar{A} \leftrightarrow \bar{A}A) \right ] \nonumber \\
L_{\text{pot}} &= i\epsilon^{ABCD} \phi_A \bar{\psi}_B \phi_C \bar{\psi}_D - i\epsilon_{ABCD} \bar{\phi}^A \psi^B \bar{\phi}^C \psi^D \nonumber \\
&+ \bar{\phi}^A \phi_A \bar{\psi}_B \psi^B - \phi_A \bar{\phi}^A \psi^B \bar{\psi}_B + 2\bar{\phi}^A \psi^B \bar{\psi}_A \phi_B - 2\phi_A \bar{\psi}_B \psi^A \bar{\phi}^B \nonumber \\
&- \frac{1}{3} (\phi_A \bar{\phi}^A)^3 - \frac{1}{3} (\bar{\phi}^A \phi_A)^3 - \frac{4}{3} \phi_A \bar{\phi}^C \phi_B \bar{\phi}^A \phi_C \bar{\phi}^B + 2\phi_A \bar{\phi}^A \phi_B \bar{\phi}^C \phi_C \bar{\phi}^B. \nonumber
\end{align}
The covariant derivatives are defined to act as
\begin{align}
D_\mu \zeta &= \partial_\mu \zeta + i(A\bar{A}_\mu \zeta - \zeta \bar{A}A_\mu) \nonumber \\
D_\mu \bar{\zeta} &= \partial_\mu \bar{\zeta} + i(\bar{A}A_\mu \bar{\zeta} - \bar{\zeta} A\bar{A}_\mu). \label{d-action}
\end{align}
The cohomology problem we will study first becomes interesting at one loop. For this reason, we will not be allowed to trade covariant derivatives for ordinary derivatives. More precisely, the terms in the Lagrangian which matter at this order are the ones up to $O(k^{-1})$ when all fields are rescaled by $k^{-1/2}$.

It is worth noting that 3d gauge theories also include local operators which cannot be built out of elementary fields in the Lagrangian. These operators (which are crucial for establishing the SUSY enhancement mentioned above \cite{1007.4861}) are called \textit{monopole operators} because they carry topological charge. Monopole operators would greatly complicate constructive explorations of the Hilbert space at strong coupling but it is easy to argue that they are asymptotically heavy at weak coupling. As reviewed in \cite{2110.01335}, the singular field configurations creating bare monopoles have gauge charges which are quantized in units of $k$. To restore gauge invariance, one must dress them with an $O(k)$ number of matter fields.\footnote{One can also see that these operators are heavy in the large-$k$ limit from the geometric perspective. Treating $S^7$ in the M theory background as a circle fibration of $\mathbb{CP}^3$ in the type IIA background, monopoles are dual to Kaluza-Klein modes on the fiber circle which has radius $k^{-1}$.} We will work at weak coupling and therefore neglect monopole operators from now on. See however \cite{bsvz25} for insight into how they affect the fortuity problem.\footnote{Without constructive enumeration, fortuity can also be studied at the level of the strongly-coupled index (which differs from the weakly-coupled one because of monopole operators). See \cite{hlt25} for recent progress in a family of theories which includes the $k = 1$ case of ABJM.}

The calculations that follow will analyze ABJM theory using the superconformal algebra $\mathfrak{osp}(\mathcal{N}|4)$. Our conventions for it are provided in Appendix \ref{app:superconformal-algebra}. In particular, our supercharges $Q_{\alpha r}$ and $S^{\alpha}_{\phantom\alpha r}$ are in the vector representation of $\mathfrak{so}({\cal N})$. This is to be contrasted with \cite{kllo25,bsvz25} which use the R-symmetry algebra $\mathfrak{su}(4)\simeq \mathfrak{so}(6)$ and hence have supercharges in the anti-symmetric tensor representation.
%We choose this convention as it is more uniform in regards to ${\cal N}$. In other words, for any ${\cal N}$ the R-symmetry algebra is $\mathfrak{so}(n)$, whereas the isomorphism $\mathfrak{so}(6)\simeq \mathfrak{su}(4)$ is specific to ${\cal N}=6$.
As for different values of ${\cal N}$, we note that ABJM with $k=1,2$ has ${\cal N}=8$ while ${\cal N}=5$ appears in the context of ABJ quadrality \cite{Honda:2017nku}. See also \cite{0804.2907} for a classification of various ${\cal N} \geq 4$ Chern-Simons matter theories based on supergroups. This convention will be important when defining BPS operators.

\subsection{The supercharge and BPS letters}

There are immediate problems with trying to pick $Q = Q_{-6}$ in analogy with the standard procedure to define $\tfrac{1}{16}$-BPS operators with 4d ${\cal N}=4$ SUSY. First,
\begin{align}
\{ Q_{\alpha r}, Q_{\beta s} \} = 2 \delta_{rs} P_{\alpha\beta} \label{3d-algebra1}
\end{align}
means that such a $Q$ would not be nilpotent. Second,
\begin{align}
\{ Q_{\alpha r}, S_s^\beta \} = 2i \left [ \delta_{rs} \left ( \delta_\alpha^\beta D + M_\alpha^{\;\;\beta} \right ) - i \delta_\alpha^\beta R_{rs} \right ] \label{3d-algebra2}
\end{align}
shows that R-symmetry would completely drop out of $\{ Q, Q^\dagger \}$. Any representation annihilated by this anti-commutator would have to lie below a unitarity bound that one can already get from the conformal algebra without SUSY. Indeed, superconformal unitarity bounds at threshold occur when the \textit{eigenvalues} of $\{ Q_{\alpha r}, S_s^\beta \}$ vanish. The vanishing of a particular diagonal entry is not significant in general.

Clearly, these conclusions do not depend on the choice $Q = Q_{-6}$. Whenever $Q$ is a real linear combination of the $Q_{\alpha r}$, its square will be too interesting and its kernel will be too boring. These problems can be traced to the fact that the fundamental of $\mathfrak{so}(\mathcal{N})$, in which the supercharges transform, is real. Fortunately, we can recall that the fundamental for $\mathcal{N} = 2$ splits into positively and negatively charged irreps which are complex \cite{cdi16}. This trick, which has previously been used in \cite{0903.4172}, allows us to define a nilpotent $Q$ and $Q^\dagger$ using the complex combinations
\begin{align}
Q \equiv Q_{-1} + i Q_{-2}, \quad Q^\dagger \equiv S_1^- - i S_2^-. \label{3d-qdef}
\end{align}
They generate an $\mathcal{N} = 2$ superconformal algebra inside the $\mathcal{N} = 6$ one. Referring to \eqref{3d-algebra2} again, BPS states are annihilated by (see \cite{d08})
\begin{align}\label{eq:bps-bound-abjm}
\{ Q, Q^\dagger \} = 4E - 4J - 2(2R_1 + R_2 + R_3)
\end{align}
where $R_1$, $R_2$ and $R_3$ are Dynkin labels. It will also be convenient to use the Cartans
\begin{align}
H_1 = R_1 + \frac{R_2 + R_3}{2}, \quad H_2 = \frac{R_2 + R_3}{2}, \quad H_3 = \frac{R_2 - R_3}{2}
\end{align}
which are charges with respect to rotation in three orthogonal planes.

States in ABJM theory whose energies do not scale with $k$ are created by gauge-invariant operators built as linear combinations of multi-traces on the fundamental fields. Demanding them to be $\frac{1}{12}$-BPS further introduces the constraint $\{Q,Q^\dagger\}{\cal O}=0$. By a well-known theorem, that we review and prove in Appendix \ref{app:bps-cohomology-isomorphism}, these operators are in one-to-one correspondence with the cohomology classes of $Q$, there being one and only one representative of each such class that obeys the BPS condition. By \eqref{eq:bps-bound-abjm}, the BPS condition relates the energy $E$ to spin and R charges which cannot change continuously. It is therefore clear that the scaling dimension of a BPS operator does not receive quantum corrections. Combined with the fact that \eqref{eq:bps-bound-abjm} is non-negative, we see that for an operator to be BPS at the quantum level, it must be built from fields which are at least BPS at the classical level. These fields are called \textit{BPS letters}. In the rest of this subsection, we will identify the BPS letters of ABJM theory and determine how the \textit{one-loop} supercharge acts on them.

%To study the cohomology problem it is convenient to build the operators out of the fundamental fields that satisfy the classical BPS condition, in which the classical scaling dimension appears, instead of the full quantum-corrected one. These fundamental fields are what we define to be the \textbf{BPS letters}.
To find explicit expressions for the BPS letters in ABJM theory, it is helpful to note that there is a ``little algebra'' $\mathfrak{so}(4) \subset \mathfrak{so}(6)$ preserving both vectors in \eqref{3d-qdef}.\footnote{This is part of the so called centralizer algebra which will be introduced in the next subsection.} In particular, consider the $4 \times 4$ Weyl matrices $\Gamma_{AB}$ and $\bar{\Gamma}^{AB}$ which comprise the $8 \times 8$ Dirac matrices. It is convenient to use conventions from \cite{bls08} and take
\begin{align}
\Gamma &= \begin{bmatrix} \sigma_2 & 0 \\ 0 & \sigma_2 \end{bmatrix}, \begin{bmatrix} i\sigma_2 & 0 \\ 0 & -i\sigma_2 \end{bmatrix}, \begin{bmatrix} 0 & i\sigma_2 \\ i\sigma_2 & 0 \end{bmatrix}, \begin{bmatrix} 0 & \textbf{1} \\ -\textbf{1} & 0 \end{bmatrix}, \begin{bmatrix} 0 & -i\sigma_1 \\ i\sigma_1 & 0 \end{bmatrix}, \begin{bmatrix} 0 & -i\sigma_3 \\ i\sigma_3 & 0 \end{bmatrix} \label{weyl-basis}
%\bar{\Gamma} &= \begin{bmatrix} \sigma_2 & 0 \\ 0 & \sigma_2 \end{bmatrix}, \begin{bmatrix} -i\sigma_2 & 0 \\ 0 & i\sigma_2 \end{bmatrix}, \begin{bmatrix} 0 & -i\sigma_2 \\ -i\sigma_2 & 0 \end{bmatrix}, \begin{bmatrix} 0 & -I \\ I & 0 \end{bmatrix}, \begin{bmatrix} 0 & -i\sigma_1 \\ i\sigma_1 & 0 \end{bmatrix}, \begin{bmatrix} 0 & -i\sigma_3 \\ i\sigma_3 & 0 \end{bmatrix}. \nonumber
\end{align}
with the other chirality determined by $\bar{\Gamma}^{AB} = -\Gamma^*_{AB}$.
The spinor representations of $R_{12}$, which generates the commutant of $\mathfrak{so}(4)$ in $\mathfrak{so}(6)$, function as chiral matrices allowing a further projection onto $\pm 1$ eigenspaces. In the basis \eqref{weyl-basis}, they take the very simple form
\begin{align}
i (\Gamma_1 \bar{\Gamma}_2)_A^{\;\;B} = \begin{bmatrix} \delta_a^b & 0 \\ 0 & -\delta_a^b \end{bmatrix}, \quad i (\bar{\Gamma}_1 \Gamma_2)^A_{\;\;B} = \begin{bmatrix} -\delta^a_b & 0 \\ 0 & \delta^a_b \end{bmatrix}
\end{align}
which means that we can write our fields as
\begin{align}
\phi_A = \begin{bmatrix} \phi^L_a \\ \phi^R_a \end{bmatrix}, \quad \bar{\phi}^A = \begin{bmatrix} \bar{\phi}_L^a \\ \bar{\phi}_R^a \end{bmatrix}, \quad \psi^A_\alpha = \begin{bmatrix} \psi_{L\alpha}^a \\ \psi_{R\alpha}^a \end{bmatrix}, \quad \bar{\psi}_{A\alpha} = \begin{bmatrix} \bar{\psi}^L_{a\alpha} \\ \bar{\psi}^R_{a\alpha} \end{bmatrix}
\end{align}
for $a = 1, 2$ and conclude that $\phi^L_a, \bar{\psi}^L_{a +}, \bar{\phi}_R^a, \psi_{R +}^a$ are BPS. In addition to dropping the $L$, $R$ and $\alpha = +$ labels from now on, we will also rescale fields for later convenience and define our BPS letters to be
\begin{align}
\phi_a \equiv 2 \phi^L_a, \quad \bar{\psi}_a \equiv -2i \bar{\psi}^L_{a +}, \quad \bar{\phi}^a \equiv 2 \bar{\phi}_R^a, \quad \psi^a \equiv -2i \psi_{R +}^a. \label{bps-matter}
\end{align}
The covariant derivative in \eqref{d-action} also leads to a BPS letter. We take this to be
\begin{align}
D \equiv iD_{++}.
\end{align}

To study the $Q$-cohomology, we further need the SUSY transformations of the BPS letters. Since $Q$ is defined by \eqref{3d-qdef}, these involve the matrices
\begin{align}
(\Gamma_1 + i\Gamma_2)_{AB} = \begin{bmatrix} 0 & 0 \\ 0 & -2i \epsilon_{ab} \end{bmatrix}, \quad
(\bar{\Gamma}_1 + i\bar{\Gamma}_2)^{AB} = \begin{bmatrix} -2i\epsilon^{ab} & 0 \\ 0 & 0 \end{bmatrix}. \label{weyl-combos}
\end{align}
They also involve Dirac matrices for the \textit{spacetime} rotations which we take to be
\begin{align}
\gamma_0 = i\sigma_2, \quad \gamma_1 = \sigma_1, \quad \gamma_2 = \sigma_3.
\end{align}
In this convention, $\gamma_0$ is the matrix that appears when building Lorentz scalars. It is now possible to take the known transformation laws
\begin{equation}
\begin{split}
\delta \phi_A &= i(\Gamma_r)_{AB} \varepsilon^{r\alpha} (\gamma_0)_\alpha^{\;\;\beta} \psi^B_\beta \\
\delta \bar{\phi}^A &= -i(\bar{\Gamma}_r)^{AB} \varepsilon^{r\alpha} (\gamma_0)_\alpha^{\;\;\beta} \bar{\psi}_{B\beta} \\
\delta \psi^A_\alpha &= -(\bar{\Gamma}_r)^{AB} (\gamma^\mu)_\alpha^{\;\;\beta} \varepsilon^r_\beta D_\mu \phi_B + (\bar{\Gamma}_r)^{AB} \varepsilon^r_\alpha (\phi_C \bar{\phi}^C \phi_B - \phi_B \bar{\phi}^C \phi_C) - 2(\bar{\Gamma}_r)^{BC} \varepsilon^r_\alpha \phi_B \bar{\phi}^A \phi_C \\
\delta \bar{\psi}_{A\alpha} &= (\Gamma_r)_{AB} (\gamma^\mu)_\alpha^{\;\;\beta} \varepsilon^r_\beta D_\mu \bar{\phi}^B + (\Gamma_r)_{AB} \varepsilon^r_\alpha (\bar{\phi}^C \phi_C \bar{\phi}^B - \bar{\phi}^B \phi_C \bar{\phi}^C) - 2(\Gamma_r)_{BC} \varepsilon^r_\alpha \bar{\phi}^B \phi_A \bar{\phi}^C \\
\delta A\bar{A}_\mu &= (\Gamma_r)_{AB} \varepsilon^{r\alpha} (\gamma_0 \gamma_\mu)_\alpha^{\;\;\beta} \psi^A_\beta \bar{\phi}^B + (\bar{\Gamma}_r)^{AB} \phi_B \varepsilon^{r\alpha} (\gamma_0 \gamma_\mu)_\alpha^{\;\;\beta} \bar{\psi}_{A\beta} \\
\delta \bar{A}A_\mu &= (\Gamma_r)_{AB} \bar{\phi}^B \varepsilon^{r\alpha} (\gamma_0 \gamma_\mu)_\alpha^{\;\;\beta} \psi^A_\beta + (\bar{\Gamma}_r)^{AB} \varepsilon^{r\alpha} (\gamma_0 \gamma_\mu)_\alpha^{\;\;\beta} \bar{\psi}_{A\beta} \phi_B
\end{split} \label{full-variations}
\end{equation}
and see that many terms drop out when we specialize indices to the ones in \eqref{bps-matter}. This is because the BPS superpartners in $(\Box, \bar{\Box})$ for instance have opposite chiralities whereas the \eqref{weyl-combos} matrices preserve chirality. We then get
\begin{equation}
    \begin{aligned}
	[Q, \phi_a] &= 0, &\qquad [Q, \bar{\phi}^a] &= 0 \\
        \{ Q, \psi^a \} &= \epsilon^{bc} \phi_b \bar{\phi}^a \phi_c, &\qquad \{ Q, \bar{\psi}_a \} &= \epsilon_{bc} \bar{\phi}^b \phi_a \bar{\phi}^c,\\
	[Q, A\bar{A}] &= \epsilon_{ab} \psi^a \bar{\phi}^b + \epsilon^{ab} \phi_b \bar{\psi}_a, &\qquad [Q, \bar{A}A] &= \epsilon_{ab} \bar{\phi}^b \psi^a + \epsilon^{ab} \bar{\psi}_a \phi_b
    \end{aligned}\label{q-action-abjm}
\end{equation}
by applying
\begin{align}
[ Q_{\alpha r}, \mathcal{O} \} = \frac{\delta \mathcal{O}}{\delta (\varepsilon^r \gamma_0)^\alpha} \label{3d-spacetime-gammas}
\end{align}
to the remaining terms.\footnote{The notation of \cite{bls08} would hide the $\gamma_0$ by using a bar. We are instead using bars to keep track of bifundamentals so we have written the $\gamma_0$ explicitly.} Note that we arrived at \eqref{full-variations} by taking the results of \cite{bls08} and using the identities
\begin{align}
& \bar{\psi}_A^{\intercal} \gamma_0 \varepsilon^r = (\bar{\psi}_A^{\intercal} \gamma_0 \varepsilon^r)^{\intercal} = -\varepsilon^{r \intercal} \gamma_0^{\intercal} \bar{\psi}_A = \varepsilon^{r \intercal} \gamma_0 \bar{\psi}_A \nonumber \\
& \bar{\psi}_A^{\intercal} \gamma_0 \gamma_\mu \varepsilon^r = (\bar{\psi}_A^{\intercal} \gamma_0 \gamma_\mu \varepsilon^r)^{\intercal} = -\varepsilon^{r \intercal} \gamma_\mu^{\intercal} \gamma_0^{\intercal} \bar{\psi}_A = -\varepsilon^{r \intercal} \gamma_0 \gamma_\mu \bar{\psi}_A \label{useful-ids1} \\
& (\gamma^\mu)_\alpha^{\;\;\beta} (\gamma_\mu)_\gamma^{\;\;\delta} = 2 \delta^\delta_\alpha \delta^\beta_\gamma - \delta^\beta_\alpha \delta^\delta_\gamma \Longrightarrow (\gamma^\mu)_{++} (\gamma_\mu)_-^{\;\;\alpha} = 2 \epsilon_{-+} \delta_+^\alpha = -2 \delta_+^\alpha \nonumber
\end{align}
which are easy to verify.
% Note that if $\varepsilon^r$ were replaced by $\psi^A$ then these identities would no longer hold except in the $U(1) \times U(1)$ theory where everything commutes.
\begin{comment}
The BPS letters are, by construction, annihilated by the free contribution to the supercharge $Q$. They nevertheless have non-vanishing $Q$ action through the interacting terms that depend on the coupling, which are in correspondence to the quantum corrections that the dilatation operator receives in perturbation theory \cite{}. Consequently, to study the $Q$ cohomology we further need the supersymmetry transformation of the BPS letters. Relegating the calculation to Appendix \ref{app:q-action}, we find that the only non-trivial transformation of the BPS letters are:
\begin{equation}
    \begin{split}
        \{ Q, \psi^a \} &= \epsilon^{bc} \phi_b \bar{\phi}^a \phi_c,\\
\{ Q, \bar{\psi}_a \} &= \epsilon_{bc} \bar{\phi}^b \phi_a \bar{\phi}^c,\\
[Q, iA\bar{A}] &= \epsilon_{ab} \psi^a \bar{\phi}^b + \epsilon^{ab} \phi_b \bar{\psi}_a,\\
[Q, i\bar{A}A] &= \epsilon_{ab} \bar{\phi}^b \psi^a + \epsilon^{ab} \bar{\psi}_a \phi_b
    \end{split}
\end{equation}
\end{comment}

\subsection{The centralizer of the BPS condition}

The \textit{centralizer} of an element of a Lie superalgebra $\mathfrak{g}$ is the subalgebra of all elements of $\mathfrak{g}$ that (anti-)commute with the chosen element. The centralizer of multiple elements is their intersection. Choosing the $Q$ and $Q^\dagger$ elements of $\mathfrak{osp}(6|4)$, this is the subalgebra that preserves the BPS condition. As such it takes BPS states to BPS states, or equivalently, $Q$-cohomoogy classes to other $Q$-cohomology classes. In particular, BPS states must then be organized into representations of the centralizer. In our particular case, we find that the generators of the centralizer ${\cal C}(\{Q,Q^\dagger\})$ are
\begin{equation}
    \{D,M_-^{\phantom--},P_{++},K^{++}\}\cup \{R_{12}\}\cup \{R_{ij}: i,j\notin \{1,2\}\}
\end{equation}
in the bosonic sector and
\begin{equation}
    \{Q,Q^\dagger\}\cup \{Q_{+r}, S^+_{\phantom+r}:r\geq 3\}    
\end{equation}
in the fermionic sector.\footnote{Note that $\{Q,Q^\dagger\}$ here does not stand for the anti-commutator, but for the two-element set.}
The full computation of the algebra is provided in Appendix \ref{app:centralizer-calculation}.

The bosonic part of the algebra can be recognized as $\mathfrak{sl}(2)\oplus \mathfrak{u}(1)\oplus \mathfrak{so}(4)\oplus \mathfrak{u}(1)$. The $\mathfrak{sl}(2)\oplus \mathfrak{u}(1)$ generators are the combinations
\begin{equation}\label{eq:sl2-u1-gen-centralizer}
    {\cal H} = D - M_-^{\phantom - - },\quad P = \frac{1}{2}P_{++},\quad K = -\frac{1}{2}K^{++},\quad C = D+M_{-}^{\phantom - -}.
\end{equation}
where $\{{\cal H},P,K\}$ generate the $\mathfrak{sl}(2)$ and $C$ generathes the $\mathfrak{u}(1)$. The remaining $\mathfrak{so}(4)\oplus \mathfrak{u}(1)$ comes from R-symmetry, with $\{R_{ij} : i,j\neq \{1,2\}\}$ generating the $\mathfrak{so}(4)$ and $R_{12}$ generating the $\mathfrak{u}(1)$. Further computing the (anti-)commutators with the fermionic generators it turns out that $\{{\cal H},P,K,Q_{+r},S^+_{\phantom+ r}, R_{rs} : r,s\geq 3\}$ generates the superconformal algebra $\mathfrak{osp}(4|2)$ while $\{Q,Q^\dagger,C,R_{12}\}$ generates a $\mathfrak{u}(1|1)$ algebra. Altogether we find that the centralizer is given by
\begin{equation}
    {\cal C}(\{Q,Q^\dagger\}) \simeq \mathfrak{osp}(4|2)\oplus \mathfrak{u}(1|1).
\end{equation}
Two comments are in order. First, recent discussions in \cite{kllo25,bsvz25} have focused on $\mathfrak{osp}(4|2)$ as the interesting part of the centralizer but the $C$ and $R_{12}$ generators of $\mathfrak{u}(1|1)$ directly map onto the $E + J$ and $H_1$ Cartans which are also used to label BPS states. Fugacities for them will be introduced when we turn to the BPS partition function. Second, unlike with $\mathfrak{psu}(1,2|3)$ which appears in the centralizer for $\mathcal{N} = 4$ SYM \cite{2306.04673}, $\mathfrak{osp}(4|2)$ is a lower-dimensional superconformal algebra.

\subsection{Comparison to $\mathcal{N} = 4$ SYM}

Building up the BPS spectrum from a restricted set of letters is an approach which has been widely studied in $\mathcal{N} = 4$ SYM \cite{0803.4183,1305.6314}. We will now review this story, not only because it provides context, but also because we will see a sharper connection to ABJM theory later on when we consider composite operators.

Supercharges in 4d transform in a reducible representation of the Lorentz algebra, so they come in two types --- left and right handed supercharges which sastisfy
\begin{align}
\{ Q^r_\alpha, \bar{Q}_{\dot{\alpha} s} \} = 2 \delta^r_s P_{\alpha \dot{\alpha}}
\end{align}
and anti-commute among themselves. We can therefore choose the nilpotent supercharge $Q = Q_{-6}$ without loss of generality. From
\begin{align}
\{ Q^r_\alpha, S_s^\beta \} = 2 \delta^r_s \left ( \delta^\beta_\alpha D + 2M_\alpha^{\;\;\beta} \right ) - 4 \delta^\beta_\alpha R^r_{\;\;s},
\end{align}
it becomes clear that
\begin{align}
\{ Q, Q^\dagger \} = 2E - 4J - (R_1 + 2R_2 + 3R_3).
\end{align}
The Dynkin labels for $\mathfrak{su}(4)$ are a permutation of those for $\mathfrak{so}(6)$ and map to Cartans via
\begin{align}
H_1 = R_2 + \frac{R_1 + R_3}{2}, \quad H_2 = \frac{R_3 + R_1}{2}, \quad H_3 = \frac{R_3 - R_1}{2}.
\end{align}

The fields in the Lagrangian are: covariant derivatives $D_{\alpha \dot{\alpha}}$ in the singlet of $\mathfrak{su}(4)$, left-handed spinors $\Psi_{\alpha r}$ in the fundamental, right-handed spinors $\bar{\Psi}^{\dot{\alpha} r}$ in the anti-fundamental and scalars $\Phi^{rs}$ transforming as anti-symmetric tensors. Going through all components, the ones which satisfy the BPS bound at zero coupling are
\begin{align}
\phi^n \equiv \Phi^{4n}, \quad \psi_n \equiv -i\Psi_{n+}, \quad \lambda_{\dot{\alpha}} \equiv \bar{\Psi}^4_{\dot{\alpha}}, \quad D_{\dot{\alpha}} \equiv D_{+ \dot{\alpha}}. \label{n4-bps-letters}
\end{align}
Instead of anti-symmetrized covariant derivatives, one can also use the field strength $f$ in
\begin{align}
[D_{\dot{\alpha}}, D_{\dot{\beta}}] = \epsilon_{\dot{\alpha} \dot{\beta}} f \equiv -i \epsilon_{\dot{\alpha} \dot{\beta}} F_{++} \label{n4-bps-f}
\end{align}
as a BPS letter.\footnote{It should be no surprise that \eqref{n4-bps-letters} fall into representations of $\mathfrak{su}(3)$. This is what stabilizes an $\mathfrak{su}(4)$ vector and it appears in the centralizer identified by \cite{2306.04673}.} Computing SUSY variations leads to the crucial relations
\begin{align}
\{ Q, \lambda_{\dot{\alpha}} \} &= 0 \quad [Q, \phi^n] = 0, \quad \{ Q, \psi_m \} = -i \epsilon_{m n p} [\phi^n, \phi^p], \quad [Q, f] = i[\phi^n, \psi_n] \nonumber \\
[ Q, D_{\dot{\alpha}} \zeta \} &= -i[ \lambda_{\dot{\alpha}}, \zeta \} + D_{\dot{\alpha}} [ Q, \zeta \}. \label{n4-qactions}
\end{align}
Although these look different from \eqref{q-action-abjm}, we will see that they can be mimicked within ABJM theory in a useful way.

\subsection{The question of renormalization}

The supercharge actions written above were read off from the SUSY variations of fields in a Lagrangian. If we normalize BPS letters to make them have a canonical kinetic term, the coupling dependence is schematically
\begin{align}
[Q, \mathcal{O}\} \sim g_{\text{YM}} \mathcal{O}' \;\; (\text{SYM}), \quad\quad\quad [Q, \mathcal{O}\} \sim k^{-1} \mathcal{O}' \;\; (\text{ABJM}). \label{q-coupling}
\end{align}
As expected, these vanish in the free theory which allows any product of BPS letters to still be BPS. As we deform away from the free theory, $\Delta = 2 \{ Q, Q^\dagger \}$ tells us that unprotected operators may receive anomalous dimensions starting at $O(g_{\text{YM}}^2)$ in SYM (one loop) and $O(k^{-2})$ in ABJM (two loops). It is well known that the leading quantum correction away from the free theory indeed lifts many operators above the BPS bound in such a way that the superconformal index is still preserved. If one wishes to learn about quantum gravity, it is crucial to understand whether operators vanishing under \eqref{q-coupling} remain BPS at higher orders. Conservatively, all we can say is that the number of loop orders which need to be checked in bounded by the number of BPS letters \cite{2306.01039}.

The strongest claim one might make is that all BPS operators are protected except for the ones which are lifted by the first quantum correction. Going back to \cite{0803.4183}, this has been conjectured many times for $\mathcal{N} = 4$ SYM but recent work has reached the opposite conclusion. First, the one-loop BPS partition functions for ${\rm SO(7)}$ SYM and ${\rm USp(6)}$ SYM were compared in \cite{cl25a}. In two different charge sectors, it was found that the enumeration for ${\rm SO(7)}$ produced one more cohomology class compared to ${\rm USp(6)}$. This is a sign of operator renormalization since the weak-coupling spectrum of one theory is believed to be the strong-coupling spectrum of the other by S-duality. Second, \cite{cl25b} showed that a fortuitous class found in \cite{glrt25} recombines with a monotone class as a result of the Konishi anomaly. This is a well known mechanism which prevents supercharges from satisfying the Leibniz rule in perturbation theory. Most recently, the $O(g_{\text{YM}}^2)$ correction to the supercharge of 4d gauge theories was computed in \cite{bk25} and found to take a simple form.

Although definitive checks have not yet been done for ABJM theory, we expect that it will also turn out to have classically BPS operators with anomalous dimensions starting at various powers of $k^{-1}$. In particular, we are skeptical of the conjecture made in \cite{bsvz25} that operators killed by \eqref{q-coupling} other than monopoles are not renormalized. As such, our statements in the following section about certain operators being fortuitous should be seen as provisional. On the other hand, the results of \cite{cl25a} suggest that higher-loop corrections which affect BPS counting might be relatively uncommon --- the discrepancy for monotone classes was seen for instance because ${\rm USp(6)}$ had 84 while ${\rm SO(7)}$ had 85. It will be interesting to see if the empirically observed sparsity becomes more pronounced as $N$ is increased.

\section{Two Types of Searches}\label{sec:abjm-details}

The transformation rules for BPS letters under $Q$ form the starting point for more detailed studies of the BPS spectrum. A computational task one can always do is proceeding level-by-level in $E + J$, which is preserved by $Q$, to determine which cohomology classes are monotone and which are fortuitous. This is a brute force search in the spirit of \cite{2209.06728} which will be done at the start of this section. As anticipated in \cite{bsvz25}, we will be able to find several fortuitous classes in ABJM theory on a laptop for small values of $N$. This contrasts nicely with the $\mathcal{N} = 4$ SYM case which required significant effort to find even the first fortuitous class for ${\rm SU}(2)$.

The second part of this section will focus on the fact that fortuitous cohomology classes have been found in other ways as well. In $\mathcal{N} = 4$ SYM, there are special subsectors enabling easier calculations \cite{2304.10155,2312.16443,2412.08695,gklm25} and these have led to infinite families. Our main result here will be that ABJM theory admits subsectors where $Q$ acts in the same way, thereby enabling a similar construction of infinite families. The only way for these to not be fortuitous is for some (as yet unknown) principle to make them $Q$-exact.

\subsection{Brute force enumeration}

After fixing a maximum value for $E + J$, the task now is to find the dimensions of two cohomologies --- the one for $Q$ restricted to \textit{general multi-trace operators} below the cutoff and the one for $Q$ restricted to \textit{multi-graviton operators} below the cutoff. When operators are constructed using an overcomplete generating set, this problem amounts to finding the rank of a large matrix. Due to the presence of four independent conserved charges, this matrix is block diagonal. One can therefore run the computation for different charge sectors in parallel. Charges of composite operators are assembled from those of the BPS letters which are shown in Table \ref{charge-table}.
\begin{table}[h]
\centering
\begin{tabular}{|c|ccccc|}
\hline
& $E$ & $J$ & $H_1$ & $H_2$ & $H_3$ \\
\hline
$Q$ & $\frac{1}{2}$ & $-\frac{1}{2}$ & $1$ & $0$ & $0$ \\
$\phi_a$ & $\frac{1}{2}$ & $0$ & $\frac{1}{2}$ & $\pm \frac{1}{2}$ & $\mp \frac{1}{2}$ \\
$\bar{\phi}^a$ & $\frac{1}{2}$ & $0$ & $\frac{1}{2}$ & $\pm \frac{1}{2}$ & $\pm \frac{1}{2}$ \\
$\psi^a$ & $1$ & $\frac{1}{2}$ & $\frac{1}{2}$ & $\pm \frac{1}{2}$ & $\pm \frac{1}{2}$ \\
$\bar{\psi}_a$ & $1$ & $\frac{1}{2}$ & $\frac{1}{2}$ & $\pm \frac{1}{2}$ & $\mp \frac{1}{2}$ \\
$D$ & $1$ & $1$ & $0$ & $0$ & $0$ \\
\hline
\end{tabular}
\caption{Cartan charges of the BPS letters under the maximally bosonic subalgebra of $\mathfrak{osp}(6|4)$. By convention, we associate the upper sign with $a = 1$ and the lower sign with $a = 2$.}
\label{charge-table}
\end{table}

The number of multi-trace operators to consider grows rapidly with $E + J$. A demonstration of this for $\mathcal{N} = 4$ SYM with gauge group ${\rm SU}(2)$ is shown in Figure \ref{n4-su2-plot}. This case has two extra charges compared to ABJM. One of them is the second Cartan of the rotation group while the other is associated with the ${\rm U}(1)$ bonus symmetry of \cite{hep-th/9811047,hep-th/9905020}.\footnote{See \cite{1406.4814,2103.15830} for bonus symmetries which have been proposed for 3d $\mathcal{N} \geq 4$ SCFTs.} Nevertheless, the charge assignments in ABJM come with two simplifying features. One of them is $E + J \in \mathbb{N}$ which follows from the need to pair barred letters with unbarred letters. The other is a $D_4$ automorphism which cuts down the required work by almost a factor of 8. Indeed, the structure of the cohomology is unchanged by orbits of the last two Cartans under
\begin{align}
H_2 \leftrightarrow -H_2, \quad H_3 \leftrightarrow -H_3, \quad H_2 \leftrightarrow H_3.
\end{align}
As actions on fields, these transformations correspond to $(\phi_1, \phi_2, \psi^1, \psi^2) \leftrightarrow (\bar{\phi}^1, \bar{\phi}^2, \bar{\psi}_1, \bar{\psi}_2)$, $(\phi_1, \phi_2, \psi^1, \psi^2) \leftrightarrow (\bar{\phi}^2, \bar{\phi}^1, \bar{\psi}_2, \bar{\psi}_1)$ and $(\phi_1, \bar{\phi}^1, \psi^1, \bar{\psi}_1) \leftrightarrow (\phi_2, \bar{\phi}^2, \psi^2, \bar{\psi}_2)$ respectively.
\begin{figure}[h]
\centering
\includegraphics[scale=0.75]{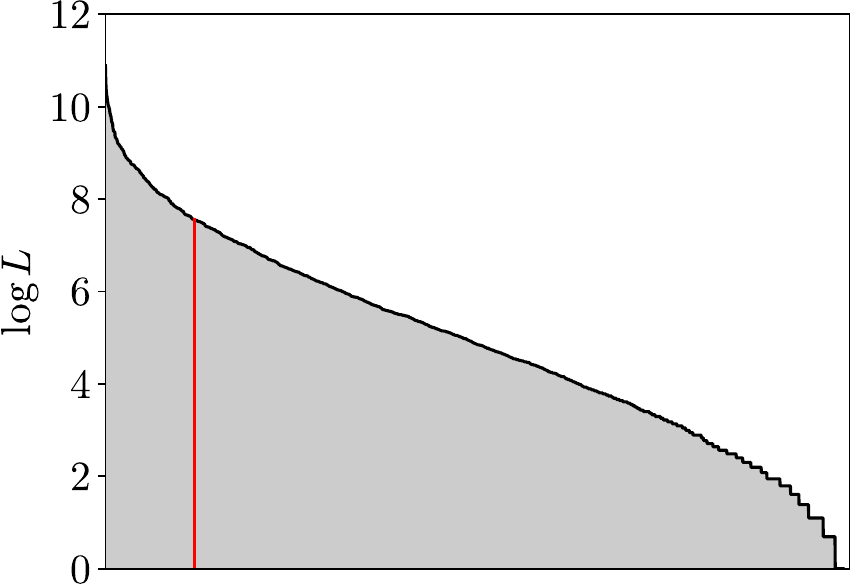}
\caption{Distribution of file sizes in the search for the simplest fortuitous cohomology class in $\mathcal{N} = 4$ SYM up to $E + J_L = 12$ --- the 15864 files are associated with the different $(Y, J_L, J_R, H_1, H_2, H_3)$ charge sectors. Within a given file, each of the $L$ lines specifies a different admissible grouping of single-trace words into a multi-trace word. The single-trace words we use are limited to length $3$ instead of $2$ for the technical reason described in the text. The longest file corresponds to $\left ( 8, \frac{3}{2}, 0, \frac{5}{2}, \frac{5}{2}, \frac{5}{2} \right )$ and has 53546 lines. The shortest files have one line and there are 173 of them. The file containing the fortuitous class, labelled in red, is $\left ( 7, \frac{5}{2}, 0, \frac{3}{2}, \frac{3}{2}, \frac{3}{2} \right )$ with 1908 lines.}
\label{n4-su2-plot}
\end{figure}

Because some charge sectors are more populated than others and many contain no fortuitous states at all, one might imagine that it is possible to make gains by performing only a limited search. This was first attempted in \cite{1305.6314} which missed the red line in Figure \ref{n4-su2-plot}. More recently, refined searches of $\mathcal{N} = 4$ SYM based on the index have been used to great effect in \cite{2304.10155,2312.16443,2412.08695,gklm25}. Postponing studies of this nature to future work, we will simply visit all charge sectors allowed by the cutoff to be sure we do not miss fortuitous states in ABJM. Our calculations follow the scheme outlined in \cite{2209.06728} and are implemented in three scripts attached to this paper's arXiv submission.

\textbf{ABJM\_step1.py}. This script accepts two parameters: \texttt{N} and \texttt{cutoff}. It is responsible for constructing all of the necessary multi-trace words. Explicit expressions for single-trace words are output line-by-line into a file called \texttt{st\_words}. Those that are $Q$-closed (single gravitons) are similarly expressed in a file called \texttt{sg\_words}.\footnote{To identify these, it is very important to remove formal traces which happen to vanish for all $N$. An example is $\text{Tr}(\psi^1 \bar{\phi}^1 \psi^1 \bar{\phi}^1)$. This goes to minus itself under cyclicity because the $N^2$ entries of $\psi^1$ anti-commute.} Multi-trace words for the various charge sectors are then stored in files with names of the form \texttt{2J\_2H1\_2H2\_2H3\_mt} and \texttt{2J\_2H1\_2H2\_2H3\_mg} so as to avoid half-integers. The suffix \texttt{mt} indicates ordinary multi-trace words while \texttt{mg} is used for multi-gravitons. Within these files, lines are tuples of integers which refer to line numbers of either \texttt{st\_words} or \texttt{sg\_words} as appropriate. An important optimization is that certain trace relations are accounted for right away. When we form multi-trace words of $N \times N$ matrices in all possible ways, it is redundant to use single-trace words with more than $N$ matrices. Since barred and unbarred letters alternate in ABJM, it is really pairs of letters which count as matrices in this context. For this reason, we only work with single-graviton words having $2N$ letters or fewer. With ordinary single-trace words, we allow the length to reach $2N + 2$ since we need to be able to compute the kernel of $Q$ which creates two extra letters.

\textbf{ABJM\_step2.jl}. This script accepts a string of the form \texttt{2J\_2H1\_2H2\_2H3} which we refer to as twice the charge vector $\textbf{q}$. It outputs three matrices in sparse CSR format which capture the effects of the remaining trace relations: \texttt{2J\_2H1\_2H2\_2H3\_mt.jls}, \texttt{2J\_2H1\_2H2\_2H3\_Q\_mt.jls} and \texttt{2J\_2H1\_2H2\_2H3\_mg.jls}. To form the first matrix, we read the \texttt{mt} file corresponding to $\textbf{q}$, expand the traces into component form and finally sort all field products into lexicographic order while recalling that fermions anti-commute. For the second matrix, we instead read the \texttt{mt} file corresponding to $\textbf{q} - \textbf{q}_0 = \textbf{q} - (-1, 2, 0, 0)$. We then act with $Q$ on these words before expanding the traces and sorting. To form the third matrix, we simply append rows to the second. These rows are sorted and expanded traces which come from reading the \texttt{mg} file corresponding to $\textbf{q}$. The three matrices are said to have row spaces
\begin{align}
V_{\textbf{q}}, \quad QV_{\textbf{q} - \textbf{q}_0}, \quad \text{span}(W_{\textbf{q}} \cup QV_{\textbf{q} - \textbf{q}_0})
\end{align}
respectively.

\textbf{ABJM\_step3.jl}. This script uses QR factorization to compute the ranks of the various matrices. Defining $n$ to be the $E + J$ cutoff, $\text{dim}(V_{\textbf{q}})$ becomes a line in \texttt{count\_before\_n}, $\text{dim}(QV_{\textbf{q} - \textbf{q}_0})$ becomes a line in \texttt{count\_after\_n} and $\text{dim}(\text{span}(W_{\textbf{q}} \cup QV_{\textbf{q} - \textbf{q}_0}))$ becomes a line in \texttt{count\_with\_n}. It is then straightforward to verify that the contribution of $\textbf{q}$ to $\text{dim}(H^{\ast}_{Q_N}(\mathcal{H}_N))$ is
\begin{align}
\text{dim}(V_{\textbf{q}}) - \text{dim}(QV_{\textbf{q}}) - \text{dim}(QV_{\textbf{q} - \textbf{q}_0}).
\end{align}
The analogous cohomology restricted to multi-graviton states has dimension
\begin{align}
\text{dim}(\text{span}(W_{\textbf{q}} \cup QV_{\textbf{q} - \textbf{q}_0})) - \text{dim}(QV_{\textbf{q} - \textbf{q}_0}).
\end{align}

\subsection{Counting at low level}

The results of the algorithm described above can be packaged into the superconformal index
\begin{align}
I(x, \textbf{y}) = \text{Tr} \left [ (-1)^F x^{E + J} y_2^{H_2} y_3^{H_3} \right ] \label{i-def}
\end{align}
or the BPS partition function
\begin{align}
Z(x, \textbf{y}) = \text{Tr} \left [ x^{E + J} y_1^{H_1} y_2^{H_2} y_3^{H_3} \right ]. \label{z-def}
\end{align}
The trace must be restricted to BPS states in \eqref{z-def} but this is automatic in \eqref{i-def}. The index for ABJM theory was computed for all $k$ in \cite{0903.4172} using localization.\footnote{When $k$ is asymptotically large (allowing us to interpret it as a coupling constant and not just a level), it ceases to affect light operators in the index as expected.} This does not eliminate the need for explicit state counting because there are presently no exact results for the multi-graviton version of the index. As such, we will phrase our results in terms of the BPS partition function which can be analytically continued and unrefined to arrive at the index as described in \cite{kllo25}.

It will be convenient to use the basis of $\mathfrak{su}(2) \oplus \mathfrak{su}(2)$ Dynkin labels instead of the Cartans $H_2$ and $H_3$. This can be done by defining the fugacities
\begin{align}
y_1 = y, \quad y_2 = y_+ y_-, \quad y_3 = y_+ / y_-.
\end{align}
The natural $\mathfrak{su}(2)$ characters are then
\begin{align}
\chi_{2j}(y_\pm) = \sum_{m = -j}^j y_\pm^{2m}
\end{align}
for the spin-$j$ representation. In the fortuitous partition function $Z_N - Z_N^{\text{grav}}$, it will be convenient to factor out
\begin{align}
\chi(x, y_2, y_3) = \frac{\prod_\pm (1 + xy_2^{\pm 1}) (1 + xy_3^{\pm 1})}{1 - x^2}
\end{align}
which accounts for descendants in long multiplets of $\mathfrak{osp}(4|2)$. To do this for $Z_N$ and $Z_N^{\text{grav}}$ separately, we will need to account for short multiplets at $E + J = H_2$ --- the unitarity bound of $\mathfrak{osp}(4|2)$. The characters for these have the expression
\begin{align}
\hat{\chi}_n(x, \textbf{y}) = \sum_{k = 0}^\infty (-x)^k \chi_{n + k}(y_+) \chi_{n + k}(y_-) \chi(x, y_2, y_3)
\end{align}
which was found in \cite{1911.10391} using recombination rules. The existence of a positive expansion in these characters provides a strong check on our results.

Starting with $N = 2$, the BPS partition function has the expansion
\begin{align}
& \frac{Z_2 - \sum_n \lfloor \tfrac{n + 2}{2} \rfloor x^n y^n \hat{\chi}_n - 1}{\chi} = x^2 y^2 + x^3 \Big [ 2y^2 \chi^+_1 \chi^-_1 + y^3 (\chi^+_1 \chi^-_1 + \chi^+_1 \chi^-_3 + \chi^+_3 \chi^-_1) \Big ] \\
& \qquad + x^4 \Big [ y^2 (\chi^+_2 + \chi^-_2) + y^3 (1 + 2\chi^+_2 + 2\chi^-_2 + 4\chi^+_2 \chi^-_2) + y^4 (1 + \chi^+_4 + \chi^-_4 + \chi^+_2 \chi^-_4 + \chi^+_4 \chi^-_2 + 2 \chi^+_2 \chi^-_2) \Big ] \nonumber \\
& \qquad + x^5 \Big [ 2 y^2 \chi^+_1 \chi^-_1 + y^3 (7 \chi^+_1 \chi^-_1 + 3 \chi^+_1 \chi^-_3 + 3 \chi^+_3 \chi^-_1) + y^4 (3 \chi^+_1 \chi^-_1 + 4 \chi^+_1 \chi^-_3 + 4 \chi^+_3 \chi^-_1 + 6 \chi^+_3 \chi^-_3) \nonumber \\
& \qquad + y^5 (\chi^+_1 \chi^-_1 + 2 \chi^+_3 \chi^-_3 + \chi^+_1 \chi^-_3 + \chi^+_3 \chi^-_1 + \chi^+_1 \chi^-_5 + \chi^+_5 \chi^-_1 + 2 \chi^+_3 \chi^-_5 + 2 \chi^+_5 \chi^-_3) \Big ] + O(x^6). \nonumber
\end{align}
We can also subtract the multi-graviton partition function to find
\begin{align}
\frac{Z_2 - Z_2^{\text{grav}}}{\chi} &= x^3 y^2 \chi^+_1 \chi^-_1 + x^4 y^3 (1 + \chi^+_2 \chi^-_2) \nonumber \\
&+ x^5 \Big [ y^2 \chi^+_1 \chi^-_1 + y^3 (2 \chi^+_1 \chi^-_1 + \chi^+_1 \chi^-_3 + \chi^+_3 \chi^-_1) + y^4 \chi^+_3 \chi^-_3 \Big ] + O(x^6). \label{zfort-n2}
\end{align}
This reveals the fortuitous state found in \cite{bsvz25} at $E + J = 3$ along with six more fortuitous multiplets of $\mathfrak{osp}(4|2) \oplus \mathfrak{u}(1|1)$ at higher levels.

Looking at $N = 3$ next, we have the BPS partition function
\begin{align}
& \frac{Z_3 - \sum_n n x^n y^n \hat{\chi}_n - 1}{\chi} = x^2 y^2 + x^3 \Big [ y^2 \chi^+_1 \chi^-_1 + y^3 (2 \chi^+_1 \chi^-_1 + \chi^+_1 \chi^-_3 + \chi^+_3 \chi^-_1) \Big ] \\
& \qquad + x^4 \Big [ y^2 (\chi^+_2 + \chi^-_2) + y^3 (1 + 2\chi^+_2 + 2\chi^-_2 + 4\chi^+_2 \chi^-_2) \nonumber \\
& \qquad + y^4 (2 + 4 \chi^+_2 \chi^-_2 + \chi^+_2 + \chi^-_2 + 2 \chi^+_4 + 2 \chi^-_4 + 2 \chi^+_2 \chi^-_4 + 2 \chi^+_4 \chi^-_2) \Big ] \nonumber \\
& \qquad + x^5 \Big [ y^2 \chi^+_1 \chi^-_1 + y^3 (5 \chi^+_1 \chi^-_1 + 2 \chi^+_1 \chi^-_3 + 2 \chi^+_3 \chi^-_1) + 8 y^4 (\chi^+_1 \chi^-_1 + \chi^+_1 \chi^-_3 + \chi^+_3 \chi^-_1 + \chi^+_3 \chi^-_3) \nonumber \\
& \qquad + y^5 (4 \chi^+_1 \chi^-_1 + 6 \chi^+_3 \chi^-_3 + 4 \chi^+_1 \chi^-_3 + 4 \chi^+_3 \chi^-_1 + 3 \chi^+_1 \chi^-_5 + 3 \chi^+_5 \chi^-_1 + 4 \chi^+_3 \chi^-_5 + 4 \chi^+_5 \chi^-_3) \Big ] + O(x^6). \nonumber
\end{align}
Computing the multi-graviton partition function as well and subtracting it, the result is
\begin{align}
\frac{Z_3 - Z_3^{\text{grav}}}{\chi} &= x^4 y^3 (1 + \chi^+_2 \chi^-_2) + x^5 y^4 (2 \chi^+_1 \chi^-_1 + \chi^+_1 \chi^-_3 + \chi^+_3 \chi^-_1 + \chi^+_3 \chi^-_3) + O(x^6). \label{zfort-n3}
\end{align}
Comparing to \eqref{zfort-n2}, both primaries at $O(x^4)$ and one at $O(x^5)$ have remained intact, suggesting that the trace relations making them $Q$-closed are valid at $N = 2,3$.

Finally, the desired partition functions for $N = 4$ are
\begin{align}
& \frac{Z_4 - \sum_n \left ( n + \lfloor |\tfrac{n - 2}{2}| \rfloor \right ) \hat{\chi}_n - 1}{\chi} = x^2 y^2 + x^3 \Big [ y^2 \chi^+_1 \chi^-_1 + y^3 (2 \chi^+_1 \chi^-_1 + \chi^+_1 \chi^-_3 + \chi^+_3 \chi^-_1) \Big ] \\
& \qquad + x^4 \Big [ y^2 (\chi^+_2 + \chi^-_2) + y^3 (2\chi^+_2 + 2\chi^-_2 + 3\chi^+_2 \chi^-_2) \nonumber \\
& \qquad + y^4 (3 + 5 \chi^+_2 \chi^-_2 + \chi^+_2 + \chi^-_2 + 2 \chi^+_4 + 2 \chi^-_4 + 2 \chi^+_2 \chi^-_4 + 2 \chi^+_4 \chi^-_2) \Big ] \nonumber \\
& \qquad + x^5 \Big [ y^2 \chi^+_1 \chi^-_1 + y^3 (5 \chi^+_1 \chi^-_1 + 2 \chi^+_1 \chi^-_3 + 2 \chi^+_3 \chi^-_1) + y^4 (7 \chi^+_1 \chi^-_1 + 7 \chi^+_1 \chi^-_3 + 7 \chi^+_3 \chi^-_1 + 8 \chi^+_3 \chi^-_3) \nonumber \\
& \qquad + y^5 (6 \chi^+_1 \chi^-_1 + 8 \chi^+_3 \chi^-_3 + 6 \chi^+_1 \chi^-_3 + 6 \chi^+_3 \chi^-_1 + 4 \chi^+_1 \chi^-_5 + 4 \chi^+_5 \chi^-_1 + 5 \chi^+_3 \chi^-_5 + 5 \chi^+_5 \chi^-_3) \Big ] + O(x^6). \nonumber
\end{align}
and
\begin{align}
\frac{Z_4 - Z_4^{\text{grav}}}{\chi} = x^5 y^4 (\chi^+_1 \chi^-_1 + \chi^+_3 \chi^-_3) + O(x^6). \label{zfort-n4}
\end{align}

\subsection{Explicit representatives}

Taking the viewpoint that fortuitous states are dual to extremal black holes, it is highly desirable to construct them microscopically. The partition functions $Z_N - Z_N^{\text{grav}}$ narrow down which letters should be used when doing so. Looking at the $N = 2$ result \eqref{zfort-n2} for instance, the requirements of $E = \frac{5}{2}$, $J = \frac{1}{2}$ and $H_1 = 2$ in the $O(x^3 y^2)$ term can only be satisfied with one fermion and three bosons. Their indices should be contracted in a way which produces a doublet under each $\mathfrak{su}(2)$ in the $\mathfrak{osp}(4|2)$ R-symmetry. At the level of cohomology, this state was constructed in \cite{bsvz25}. The result was given in \eqref{abjm-fort} which we restate here for convenience:
\begin{align}
\mathcal{O}^a_b &= \epsilon^{cd} \left [ 4 \text{Tr}(\bar{\psi}_c \phi_d \bar{\phi}^a \phi_b) - 3 \text{Tr}(\bar{\psi}_c \phi_d) \text{Tr}(\bar{\phi}^a \phi_b) \right ] - \epsilon_{cd} \left [ 4 \text{Tr}(\psi^c \bar{\phi}^d \phi_b \bar{\phi}^a) - 3 \text{Tr}(\psi^c \bar{\phi}^d) \text{Tr}(\phi_b \bar{\phi}^a) \right ]. \label{abjm-fort2}
\end{align}
We will now extend this construction to two additional fortuitous classes using a procedure which easily verifies that \eqref{abjm-fort2} is correct.

The idea is to take increasingly general linear combinations of multi-traces with the right quantum numbers until the kernel of $Q$ first exhibits $N$-dependence. Specifically, there should be a state in the  covering Hilbert space $\tilde{\mathcal{H}}$ whose $Q$ action is in $I_N$ but not $I_{N + 1}$. Starting with the most general linear combination right away would not be instructive because cohomologies are ambiguous up to the addition of $Q$-exact terms and fortuitous cohomologies are ambiguous up to the addition of monotone terms. As an example, we can look at \eqref{zfort-n3} which tells us to expect an R-symmetry singlet with $E = \frac{7}{2}$, $J = \frac{1}{2}$ and $H_1 = 3$. This narrows down our search space to multi-traces with one fermion and five bosons. Within this space,
\begin{align}
\mathcal{O} &= \epsilon^{ab} \epsilon_{cd} \epsilon^{ef} \left [ 12 \text{Tr}(\bar{\psi}_a \phi_b \bar{\phi}^c \phi_e \bar{\phi}^d \phi_f) - 8 \text{Tr}(\bar{\psi}_a \phi_b \bar{\phi}^c \phi_e) \text{Tr}(\bar{\phi}^d \phi_f) + 3 \text{Tr}(\bar{\psi}_a \phi_b) \text{Tr}(\bar{\phi}^c \phi_e) \text{Tr}(\bar{\phi}^d \phi_f) \right ] \nonumber \\
&- \epsilon_{ab} \epsilon^{cd} \epsilon_{ef} \left [ 12 \text{Tr}(\psi^a \bar{\phi}^b \phi_c \bar{\phi}^e \phi_d \bar{\phi}^f) - 8 \text{Tr}(\psi^a \bar{\phi}^b \phi_c \bar{\phi}^e) \text{Tr}(\phi_d \bar{\phi}^f) + 3 \text{Tr}(\psi^a \bar{\phi}^b) \text{Tr}(\phi_c \bar{\phi}^e) \text{Tr}(\phi_d \bar{\phi}^f) \right ] \label{abjm-fort-singlet}
\end{align}
is found to be fortuitous. If only five of the six multi-traces in \eqref{abjm-fort-singlet} were used, no linear combination of them would be $Q$-closed for $N = 3$ but not $N = 4$. Using all six (but not the much larger set of multi-traces we could have considered) is what allowed us to find one solution up to normalization. It remains to verify that \eqref{abjm-fort-singlet} is not $Q$-exact. For this, we have simply checked that it lies outside the span of the basis elements for $QV_{(1,2,0,0)}$ which were found as part of the BPS counting.

We can repeat this exercise for the other $O(x^4 y^3)$ term in \eqref{zfort-n3} which is a triplet under both $\mathfrak{su}(2)$ factors. This time, we have needed to include two basis elements in which the fermions are not contracted with their immediate neighbours. Comparing the $N = 3$ and $N = 4$ kernels of $Q$ leads to the solution
\begin{align}
\mathcal{O}^{ab}_{cd} &= \epsilon^{ef} \Big [ 3 \text{Tr}(\bar{\psi}_e \phi_{\{c} \bar{\phi}^{\{a} \phi_{|f|} \bar{\phi}^{b\}} \phi_{d\}}) + 6 \text{Tr}(\bar{\psi}_e \phi_f \bar{\phi}^{\{a} \phi_{\{c} \bar{\phi}^{b\}} \phi_{d\}}) - 3 \text{Tr}(\bar{\psi}_e \phi_f) \text{Tr}(\bar{\phi}^{\{a} \phi_{\{c} \bar{\phi}^{b\}} \phi_{d\}}) \nonumber \\
&- 8 \text{Tr}(\bar{\psi}_e \phi_f \bar{\phi}^{\{a} \phi_{\{c}) \text{Tr}(\bar{\phi}^{b\}} \phi_{d\}}) + 3 \text{Tr}(\bar{\psi}_e \phi_f) \text{Tr}(\bar{\phi}^{\{a} \phi_{\{c}) \text{Tr}(\bar{\phi}^{b\}} \phi_{d\}}) \Big ] \nonumber \\
&- \epsilon_{ef} \Big [ 3 \text{Tr}(\psi^e \bar{\phi}^{\{a} \phi_{\{c} \bar{\phi}^{|f|} \phi_{d\}} \bar{\phi}^{b\}}) + 6 \text{Tr}(\psi^e \bar{\phi}^f \phi_{\{c} \bar{\phi}^{\{a} \phi_{d\}} \bar{\phi}^{b\}}) - 3 \text{Tr}(\psi^e \bar{\phi}^f) \text{Tr}(\phi_{\{c} \bar{\phi}^{\{a} \phi_{d\}} \bar{\phi}^{b\}}) \nonumber \\
&- 8 \text{Tr}(\psi^e \bar{\phi}^f \phi_{\{c} \bar{\phi}^{\{a}) \text{Tr}(\phi_{d\}} \bar{\phi}^{b\}}) + 3 \text{Tr}(\psi^e \bar{\phi}^f) \text{Tr}(\phi_{\{c} \bar{\phi}^{\{a}) \text{Tr}(\phi_{d\}} \bar{\phi}^{b\}}) \Big ]. \label{abjm-fort-triplet}
\end{align}
Checking that all components lie outside the image of $Q$ requires a construction of $QV_{(1,2,0,0)}$, $QV_{(1,2,0,\pm 1)}$, $QV_{(1,2,\pm 1,0)}$ and $QV_{(1,2,\pm 1, \pm 1)}$ but it is clear by symmetry that these checks must either all succeed or all fail. In this case, they all succeed, indicating that \eqref{abjm-fort-triplet} is not $Q$-exact. Another thing we have checked is that the $Q$-closedness of \eqref{abjm-fort-singlet} and \eqref{abjm-fort-triplet} holds not just for $N = 3$ but for $N = 2$ as well. Of course for $N = 2$, one also has the option of rewriting a trace of six letters as a multi-trace.

Let us stress again that the expressions \eqref{abjm-fort2}, \eqref{abjm-fort-singlet} and \eqref{abjm-fort-triplet} are by no means unique. They simply belong to the same cohomolgy classes as fortuitous $\mathfrak{osp}(4|2)$ primaries in ${\rm U}(2)_k \times {\rm U}(2)_{-k}$ and ${\rm U}(3)_k \times {\rm U}(3)_{-k}$ ABJM theory. To find canonical representatives, one should work out the actual BPS states annihilated by both $Q$ and $Q^\dagger$. These diagonalize the Hamiltonian
\begin{align}
H = \{Q, Q^\dagger\} + 4J + 2(2R_1 + R_2 + R_3) \label{ham-2loop}
\end{align}
which computes anomalous dimensions.\footnote{The matrices we have computed in this work are not enough to determine the Hamiltonian. One also needs the two-point functions of all possible monomials in the BPS letters \cite{2306.04673}.} Within the zero eigenspace (which is where the BPS states live), one can construct an orthonormal basis and then rotate it to remove fortuity from the maximal number of operators. This is the approach which \cite{bmv23,2306.04673} applied to $\mathcal{N} = 4$ SYM with gauge group ${\rm SU}(2)$. In this case, comparisons to previous literature made it desirable to write their fortuitous state as the sum of \eqref{sym-fort} and a $Q$-exact term. In general however, the problem of finding the ``nicest'' representative is not well defined.

\subsection{More representatives from $\mathcal{N} = 4$ SYM}

The fortuitous cohomologies we have constructed so far have at most six letters, as can be seen from \eqref{abjm-fort-singlet} and \eqref{abjm-fort-triplet}. A reasonable next step is to also find cohomologies where the number of letters is unbounded. This was first achieved for $\mathcal{N} = 4$ SYM in \cite{2304.10155} using special properties of operators in which only $\phi^n$, $\psi_n$ and $f$ appear. The advantage of this truncation is that it coincides with the subsector in \cite{hep-th/0306054} which can be described by the BMN matrix model \cite{hep-th/0202021}. This allows the monotone index $I_2^{\text{grav}}(x, \textbf{y})$ to be computed exactly. Although the BMN matrix model is only equivalent to a subsector of $\mathcal{N} = 4$ SYM at the classical level, this is enough for studying the cohomology of a supercharge which can be derived from a classical action (such as the one we are using).

The result of \cite{2304.10155} is that \eqref{sym-fort} with $E + J_L = 12$ is merely the first fortuitous class in the infinite family
\begin{align}
&\mathcal{O}_n = \epsilon^{m_1 m_2 m_3} \text{Tr}(f^2)^n \text{Tr}(\phi^{m_4} \psi_{m_1}) \text{Tr}(\phi^{m_5} \psi_{m_2}) \text{Tr}(\psi_{m_3} [\psi_{m_4}, \psi_{m_5}]) \nonumber \\
&+ \frac{n}{2} \epsilon^{m_1 m_2 m_3} \epsilon^{m_4 m_5 m_6} \text{Tr}(f^2)^{n - 1} \text{Tr}(f \psi_{m_1}) \text{Tr}(\phi^{m_7} \psi_{m_4}) \text{Tr}(\psi_{m_2} \psi_{m_5}) \text{Tr}(\psi_{m_7} [\psi_{m_3}, \psi_{m_6}]) \label{sym-tower} \\
&+ \frac{n(2n + 1)}{6912} \epsilon^{m_1 m_2 m_3} \epsilon^{m_4 m_5 m_6} \epsilon^{m_7 m_8 m_9} \text{Tr}(f^2)^{n - 1} \text{Tr}(\psi_{m_1} [\psi_{m_4}, \psi_{m_7}]) \text{Tr}(\psi_{m_2} [\psi_{m_5}, \psi_{m_8}]) \text{Tr}(\psi_{m_3} [\psi_{m_6}, \psi_{m_9}]) \nonumber
\end{align}
which has $E + J_L = 12 + 6n$. These are part of the ${\rm SU}(2)$ theory but they can be easily lifted to ${\rm U}(2)$ by applying the map \eqref{subtract-trace}. We will now show that the resulting expressions can be ``lifted'' once more to cohomology classes of ${\rm U}(2)_k \times {\rm U}(2)_{-k}$ ABJM theory which are likely to be fortuitous.
%Later, we will see that this leads to the notion of a BPS superfield in ABJM theory as well.

The idea is to find ABJM operators whose actions under $Q$ display the same structure seen in $\phi^n$, $\psi_n$ and $f$ --- letters with $E = 1$, $\frac{3}{2}$ and $2$ respectively. Since these adjoint fields are allowed to multiply each other on either the left or the right, it is natural that the corresponding ABJM operators should be \textit{bilinears}. At $E = 1$, all bilinears must involve one copy of $\phi_a$ and one copy of $\bar{\phi}^a$. Similarly, $E = \frac{3}{2}$ requires us to either pair $\phi_a$ with $\bar{\psi}_a$ or $\psi^a$ with $\bar{\phi}^a$. At $E = 2$, suitable bilinears can either have two fermions with no derivatives or two bosons with one derivative. Choosing coefficients carefully, we have found that the desired behaviour can be achieved with
\begin{align}
\phi^n &\mapsto \frac{1}{2} \phi_a \bar{\phi}^b (\sigma^n)_b^{\;\;a} \nonumber \\
\psi_n &\mapsto -\frac{1}{2} \left [ \phi_a \bar{\psi}_b (\epsilon \sigma_n)^{ab} + \psi^a \bar{\phi}^b (\sigma_n \epsilon)_{ab} \right ] \label{bmn-identifications} \\
f &\mapsto \frac{i}{4} \left [ 2 \psi^a \bar{\psi}_a - \phi_a D \bar{\phi}^a + D \phi_a \bar{\phi}^a \right ] \nonumber
\end{align}
where $\sigma_n$ is a Pauli matrix.\footnote{The adjoint index on $\sigma_n = \sigma^n$ can be raised and lowered at will.}
With these identifications, it is clear that $Q$ acts trivially on $\phi^n$ but checking that the other two $Q$ actions come out right requires a calculation.

To verify that $\psi_n$ in \eqref{bmn-identifications} transforms correctly, it helps to take
\begin{align}
[\sigma_m, \sigma_n] = 2i \epsilon_{mnp} \sigma_p, \quad \{ \sigma_m, \sigma_n \} = 2\delta_{mn}
\end{align}
and rewrite $\epsilon_{mnp}$ times the identity matrix purely in terms of Pauli matrices. Using this trick, we have
\begin{align}
-i \epsilon_{mnp} [\phi^n, \phi^p] &= -\frac{1}{8} \phi_a \bar{\phi}^b \phi_c \bar{\phi}^d (\sigma^n)_b^{\;\;a} (\sigma^p)_d^{\;\;e} \left [ (\sigma_m \sigma_n \sigma_p)_e^{\;\;c} - (\sigma_m \sigma_p \sigma_n)_e^{\;\;c} + (\sigma_n \sigma_p \sigma_m)_e^{\;\;c} - (\sigma_p \sigma_n \sigma_m)_e^{\;\;c} \right ] \nonumber \\
&= -\frac{1}{2} \phi_a \bar{\phi}^b \phi_c \bar{\phi}^d \left [ (\sigma_m)_b^{\;\;a} \delta^c_d - 2 (\sigma_m)_b^{\;\;c} \delta^a_d + (\sigma_m)_d^{\;\;c} \delta^a_b \right ] \nonumber \\
&= -\frac{1}{2} \left [ \epsilon_{ab} \epsilon^{ce} \phi_c(z) \bar{\phi}^a(z) \phi_d(z) \bar{\phi}^b(z) (\sigma_m)_e^{\;\;d} + \epsilon^{ab} \epsilon_{ed} \phi_a(z) \bar{\phi}^c(z) \phi_b(z) \bar{\phi}^d(z) (\sigma_m)_c^{\;\;e} \right ] \nonumber \\
&= -\frac{1}{2} \{ Q, \phi_c \bar{\psi}_d (\epsilon \sigma_m)^{cd} + \psi^c \bar{\phi}^d (\sigma_m \epsilon)_{cd} \} \nonumber \\
&= \{ Q, \psi_m \}. \label{check-part2}
\end{align}
In the second line, we have used the completeness relation
\begin{align}
(\sigma_m)_b^{\;\;a} (\sigma_m)_d^{\;\;c} = 2\delta^c_b \delta^a_d - \delta^a_b \delta^c_d \label{spinor-idr}
\end{align}
which is just like the last line of \eqref{useful-ids1} but for R-symmetry. In the third line, we have used
\begin{align}
\epsilon_{ab} \epsilon^{cd} = \delta_a^c \delta_b^d - \delta_a^d \delta_b^c. \label{epsilon-to-delta}
\end{align}

Our remaining task is to recover the expected transformation of $f$ in \eqref{bmn-identifications}. To do so, we will start from the commutator again and apply \eqref{spinor-idr} right away.
\begin{align}
i [\phi^m, \psi_m] =& -\frac{i}{4} (\sigma^m)_b^{\;\;a} \left [ \phi_a \bar{\phi}^b \phi_c \bar{\psi}_d (\epsilon \sigma_m)^{cd} + \phi_a \bar{\phi}^b \psi^c \bar{\phi}^d (\sigma_m \epsilon)_{cd} - \phi_c \bar{\psi}_d \phi_a \bar{\phi}^b (\epsilon \sigma_m)^{cd} - \psi^c \bar{\phi}^d \phi_a \bar{\phi}^b (\sigma_m \epsilon)_{cd} \right ] \nonumber \\
=& -\frac{i}{2} \epsilon_{bc} (\phi_a \bar{\phi}^b \psi^a \bar{\phi}^c + \psi^a \bar{\phi}^b \phi_a \bar{\phi}^c) + \frac{i}{2} \epsilon^{bc} (\phi_b \bar{\phi}^a \phi_c \bar{\psi}_a + \phi_b \bar{\psi}_a \phi_c \bar{\phi}^a) \nonumber \\
&+ \frac{i}{4} \epsilon_{bc} (\phi_a \bar{\phi}^a \psi^b \bar{\phi}^c - \psi^b \bar{\phi}^c \phi_a \bar{\phi}^a) + \frac{i}{4} \epsilon^{bc} (\phi_a \bar{\phi}^a \phi_b \bar{\psi}_c - \phi_b \bar{\psi}_c \phi_a \bar{\phi}^a) \nonumber \\
=& -\frac{i}{4} \epsilon_{bc} (\phi_a \bar{\phi}^b \psi^a \bar{\phi}^c + \psi^a \bar{\phi}^b \phi_a \bar{\phi}^c + \phi_a \bar{\phi}^b \psi^c \bar{\phi}^a + \psi^b \bar{\phi}^a \phi_a \bar{\phi}^c) \nonumber \\
& + \frac{i}{4} \epsilon^{bc} (\phi_b \bar{\phi}^a \phi_c \bar{\psi}_a + \phi_b \bar{\psi}_a \phi_c \bar{\phi}^a + \phi_b \bar{\phi}^a \phi_a \bar{\psi}_c + \phi_a \bar{\psi}_b \phi_c \bar{\phi}^a) \nonumber \\
=& -\frac{i}{2} (\epsilon_{bc} \psi^a \bar{\phi}^b \phi_a \bar{\phi}^c - \epsilon^{bc} \phi_b \bar{\phi}^a \phi_c \bar{\psi}_a) \nonumber \\
&- \frac{i}{4} \phi_a (\epsilon_{bc} \bar{\phi}^b \psi^c - \epsilon^{bc} \bar{\psi}_b \phi_c) \bar{\phi}^a + \frac{i}{4} \phi_a \bar{\phi}^a (\epsilon^{bc} \phi_b \bar{\psi}_c - \epsilon_{bc} \psi^b \bar{\phi}^c) \nonumber \\
&+ \frac{i}{4} (\epsilon^{bc} \phi_b \bar{\psi}_c - \epsilon_{bc} \psi^b \bar{\phi}^c) \phi_a \bar{\phi}^a - \frac{i}{4} \phi_a (\epsilon_{bc} \bar{\phi}^b \psi^c - \epsilon^{bc} \bar{\psi}_b \phi_c) \bar{\phi}^a \nonumber \\
=& \frac{i}{4} [Q, 2 \psi^a \bar{\psi}_a - \phi_a D \bar{\phi}^a + D \phi_a \bar{\phi}^a] \nonumber \\
=& [Q, f] \label{check-part3}
\end{align}
The middle line uses a basis in which the first two letters and the last two letters are never contracted with each other. The lines directly above and below are seen to match this due to the identities
\begin{align}
\delta^a_b \epsilon_{cd} = \epsilon_{bd} \delta^a_c - \epsilon_{bc} \delta^a_d, \quad \delta^b_a \epsilon^{cd} = \epsilon^{bd} \delta^c_a - \epsilon^{bc} \delta^d_a \label{fierz-id}
\end{align}
which are raised and lowered versions of \eqref{epsilon-to-delta}.

This approach allows us to start from a BPS operator in the BMN sector of ${\cal N}=4$ SYM and lift it to a BPS operator in ABJM theory. Moreover, if the operator is fortuitous in ${\cal N}=4$ SYM, it will be fortuitous in ABJM theory if it is not $Q$-exact under the ABJM supercharge action. This last point raises the question of whether this lift can ever be $Q$-exact in ABJM theory if it is not $Q$-exact in ${\cal N}=4$ SYM. One possibility for establishing this would be to show that the $Q$ action on an ABJM operator that cannot be written in terms of the bilinear letters \eqref{bmn-identifications} can never reproduce something that is. %Using a Python script to look for this kind of behavior we identified the operator
This is false, as demonstrated by
\begin{equation}
    \chi = \Tr(\psi^1 \bar{\phi}^1 \phi_1 \bar{\phi}^2) - \Tr(\bar{\psi}_2 \phi_2 \bar{\phi}^2 \phi_1) - \Tr(\bar{\psi}_1 \phi_2 \bar{\phi}^2 \phi_2),
\end{equation}
which is not made out of the letters \eqref{bmn-identifications} but nevertheless produces
\begin{equation}
    \{Q, \chi\} = \Tr(\phi^1 \phi^2 \phi^3) - \Tr(\phi^2 \phi^1 \phi^3)
\end{equation}
which is.
%While this still does not rule out the non $Q$-exactness of the lift, it illustrates that it could fail.
Since $Q$ might also take non-bilinears to trace relations made out of bilinears, there is a chance that the lift of a given $\mathcal{O}_n$ could be come exact at the sufficiently small values of $N$ which make it closed.
Altogether we leave it as a conjecture that the fortuitous operators will lift to fortuitous operators, which should be not $Q$-exact in ABJM theory and we leave further studying the lift and proving the conjecture to future work.

Under the assumption of this conjecture, it is interesting that we have found a new tower of fortuitous cohomologies without the need for an ABJM analogue of the special properties in \cite{hep-th/0306054,hep-th/0202021}. We have simply shown that the identifications \eqref{bmn-identifications} can be used to lift \eqref{sym-tower} by exploiting an algebraic similarity between the supercharge actions in 3d $\mathcal{N} = 6$ ABJM theory and 4d $\mathcal{N} = 4$ Super Yang-Mills. In the rest of this section, we will see that a few more similarities can be found by pushing this approach further.

\subsection{The superfield approach}

In ${\cal N}=4$ SYM there is an approach that allows for describing the monotone states more easily using a superfield construction \cite{1305.6314,2209.06728}. Specifically, single-trace monotone states have a compact expression in terms of a superfield. Arbitrary monotone operators are then built as arbitrary linear combinations of arbitrary products of these single-graviton states due to the Leibniz rule obeyed by the supercharge $Q$. In this section we briefly review this construction before attempting to generalize it to ABJM theory.

Recall that ${\cal N}=4$ SYM has the following BPS letters
\begin{align}
\phi^n \equiv \Phi^{4n}, \quad \psi_n \equiv -i\Psi_{n+}, \quad \lambda_{\dot{\alpha}} \equiv \bar{\Psi}^4_{\dot{\alpha}}, \quad D_{\dot{\alpha}} \equiv D_{+ \dot{\alpha}}, \quad f \equiv -iF_{++}
\end{align}
where $n = 1, 2, 3$, and $\Phi^{ij}$ are the scalars, $\Psi_{i\alpha}$ are the right-handed spinors and $\bar\Psi^i_{\dot\alpha}$ are the left-handed spinors, all of which are valued in the adjoint representation of ${\rm SU}(N)$ or ${\rm U}(N)$. These are the ``BMN letters'' discussed in the last subsection along with $\lambda_{\dot{\alpha}}$ and $D_{\dot{\alpha}}$ which carry spinor indices. Using the field strength $f$ defined in \eqref{n4-bps-f} ensures that none of the covariant derivatives $D_{\dot{\alpha}}$ need to be anti-symmetrized. Additionally, we can avoid anti-symmetrizing covariant derivatives with $\lambda_{\dot{\alpha}}$ due to the equation of motion
\begin{align}
\epsilon^{\dot{\alpha} \dot{\beta}} D_{\dot{\alpha}} \lambda_{\dot{\beta}} = [\phi^n, \psi_n].
\end{align}
%and $D_{\alpha\dot\alpha}$ is the covariant derivative with bispinor indices.
The first step in the superfield construction is then to define generating functions that repackage all of the symmetrized covariant derivatives of the BPS letters
\begin{align}
& \phi^m(z) \equiv \sum_{n = 0}^\infty \frac{(z^{\dot{\alpha}} D_{\dot{\alpha}})^n}{n!} \phi^m, \quad \psi_m(z) \equiv \sum_{n = 0}^\infty \frac{(z^{\dot{\alpha}} D_{\dot{\alpha}})^n}{n!} \psi_m, \nonumber \\
& f(z) \equiv \sum_{n = 0}^\infty \frac{(z^{\dot{\alpha}} D_{\dot{\alpha}})^n}{n!} f, \quad \lambda(z) \equiv \sum_{n = 0}^\infty \frac{(z^{\dot{\alpha}} D_{\dot{\alpha}})^n}{(n + 1)!} z^{\dot{\beta}} \lambda_{\dot{\beta}}. \label{4d-generating-functions}
\end{align}
With these in hand, we define the $\mathbb{C}^{2|3}$ superfield
\begin{align}
\Psi(z, \theta) = -i[\lambda(z) + 2\theta_n \phi^n(z) + \epsilon^{mnp} \theta_m \theta_n \psi_p(z) + 4\theta_1 \theta_2 \theta_3 f(z)], \label{4d-superfield1}
\end{align}
which obeys $\Psi(0,0)=0$. Derivatives with respect to the Grasmmann coordinates $\theta^i$ extract different BPS letters and derivatives with respect to $z^{\dot\alpha}$ select how many covariant derivatives act on those letters. The ${\cal N}=4$ SYM $Q$ actions on the BPS letters \eqref{n4-qactions} then imply that the superfield transforms as \cite{1305.6314}
\begin{equation}\label{eq:q-action-superfield}
    \{Q,\Psi(z,\theta)\} = \Psi(z,\theta)^2.
\end{equation}
Using $\Psi(z,\theta)$ it is possible to write the single-graviton operators as
\begin{equation}\label{eq:single-graviton-n=4}
\left.
\partial_{z_+}^{p_1}
\partial_{z_-}^{p_2}
\partial_{\theta_1}^{q_1}
\partial_{\theta_2}^{q_2}
\partial_{\theta_3}^{q_3}
\,
\operatorname{Tr}\!\left[
(\partial_{z_+}\Psi)^{k_1}
(\partial_{z_-}\Psi)^{k_2}
(\partial_{\theta_1}\Psi)^{m_1}
(\partial_{\theta_2}\Psi)^{m_2}
(\partial_{\theta_3}\Psi)^{m_3}
\right]
\right|_{{\cal Z}=0}
\end{equation}
where ${\cal Z} = (z,\theta)$. It is straightforward to show that \eqref{eq:q-action-superfield} implies that all of these single trace operators are $Q$-closed without using trace relations. For this, note that the quadratic $Q$ action implies that
\begin{equation}
    \left\{Q,\dfrac{\partial \Psi}{\partial{\cal Z}^A}\right\} = 2\Psi \dfrac{\partial\Psi}{\partial{\cal Z}^A}.
\end{equation}
Hence, when we set ${\cal Z}=0$, the $Q$ action on this derivative is zero. The single trace \eqref{eq:single-graviton-n=4} has a long string of single derivatives of $\Psi(z,\theta)$ with respect to the various coordinates ${\cal Z}^A$. As such, when the $Q$ action hits any of these single derivatives, it makes the contribution vanish. As anticipated, this involves no trace relations at all meaning that all instances of \eqref{eq:single-graviton-n=4} are monotone. Their products will be multi-trace monotone operators, interpreted here as multi-graviton operators.

To dispel potential confusions, we note that the BPS superfield $\Psi(z, \theta)$ for $\mathcal{N} = 4$ SYM is inherently on-shell. It is therefore unrelated to the more well known superfield formalism which is used to construct supersymmetric Lagrangians. In fact, superfields of the latter type do not exist for off-shell $\mathcal{N} = 4$ SUSY. This follows from the no-go theorems in \cite{rt82,rt83} pertaining to SYM and SUGRA theories in diverse dimension.

\subsection{An ABJM superfield for billinears}

We previously saw that the composite operators in \eqref{bmn-identifications} can be used to emulate a subsector of ${\cal N}=4$ SYM. This agreement can be extended to a subsector which includes covariant derivatives and we will show this using a BPS superfield for ABJM theory. Although our superfield will be able to accommodate arbitrarily many covariant derivatives, we will see that they need to be distributed in special ways.\footnote{In the amplitude context, \cite{1003.6120} formulated a superfield for ABJM theory which had other limitations. Specifically, only part of the R-symmetry could be made manifest.}

The first step in the superfield construction is to assemble the covariant derivatives of the BPS letters into generating functions such as it is done in equation \eqref{4d-generating-functions} for ${\cal N}=4$ SYM. There is immediately a problem in trying to do that for ABJM theory.  Consider the $Q$ action on a quantity with arbitrarily many covariant derivatives. It is simplest to start with $D^n \phi_a$.
\begin{align}
& [Q, D^n \phi_a] = (\epsilon^{bc} \phi_b \bar{\psi}_c - \epsilon_{bc} \psi^b \bar{\phi}^c) D^{n - 1} \phi_a - D^{n - 1} \phi_a (\epsilon_{bc} \bar{\phi}^b \psi^c - \epsilon^{bc} \bar{\psi}_b \phi_c) + D[Q, D^{n - 1} \phi_a] \\
&= \sum_{q = 0}^{n - 1} \binom{n}{n - 1 - q} \left [ D^q (\epsilon^{bc} \phi_b \bar{\psi}_c - \epsilon_{bc} \psi^b \bar{\phi}^c) D^{n - 1 - q} \phi_a  - D^{n - 1 - q} \phi_a D^q (\epsilon_{bc} \bar{\phi}^b \psi^c - \epsilon^{bc} \bar{\psi}_b \phi_c) \right ] \nonumber
\end{align}
The right-hand side is a double sum once we expand $D^q$ using the Leibniz rule. If we now sum up the $(zD)^n \phi_a$ terms on the left-hand side, weighted by some coefficients, we define the $Q$ action on a generating function $\phi_a(z)$. But then we would like the triple sum on the right-hand side to become a cubic polynomial in some other generating functions. Some experimentation shows that there is no way to choose coefficients for $\bar{\phi}^a, \psi^a, \bar{\psi}_a$ which make this true. We can find a way out by introducing a new generating function on which the superfield will be allowed to depend. For that matter, we define
\begin{align}
\lambda\bar{\lambda}(z) \equiv i \sum_{n = 0}^\infty \frac{z(zD)^n}{(n + 1)!} (\epsilon^{ab} \phi_a \bar{\psi}_b - \epsilon_{ab} \psi^a \bar{\phi}^b), \quad \bar{\lambda}\lambda(z) \equiv i \sum_{n = 0}^\infty \frac{z(zD)^n}{(n + 1)!} (\epsilon_{ab} \bar{\phi}^a \psi^b - \epsilon^{ab} \bar{\psi}_a \phi_b) \label{bilinears1}
\end{align}
and allow supersymmetry transformations to depend explicitly on these bilinears. The generating functions can then be defined analogously to ${\cal N}=4$ SYM.
\begin{align}
\phi_a(z) \equiv \sum_{n = 0}^\infty \frac{(zD)^n}{n!} \phi_a, \quad \bar{\phi}^a(z) \equiv \sum_{n = 0}^\infty \frac{(zD)^n}{n!} \bar{\phi}^a \label{unilinears} \\
\psi^a(z) \equiv \sum_{n = 0}^\infty \frac{(zD)^n}{n!} \psi^a, \quad \bar{\psi}_a(z) \equiv \sum_{n = 0}^\infty \frac{(zD)^n}{n!} \bar{\psi}_a \nonumber
\end{align}
The $Q$ action on the complete set of generating functions is then given by
\begin{align}
\{Q, \lambda\bar{\lambda}(z)\} &= -i \lambda\bar{\lambda}(z)^2 \label{q-actions1} \\
\{Q, \bar{\lambda}\lambda(z)\} &= -i \bar{\lambda}\lambda(z)^2 \nonumber \\
[Q, \phi_a(z)] &= -i [ \lambda\bar{\lambda}(z) \phi_a(z) - \phi_a(z) \bar{\lambda}\lambda(z) ] \nonumber \\
[Q, \bar{\phi}^a(z)] &= -i [ \bar{\lambda}\lambda(z) \bar{\phi}^a(z) - \bar{\phi}^a(z) \lambda\bar{\lambda}(z) ] \nonumber \\
\{Q, \psi^a(z)\} &= -i [ \lambda\bar{\lambda}(z) \psi_a(z) + \psi_a(z) \bar{\lambda}\lambda(z) ] + \epsilon^{bc} \phi_b(z) \bar{\phi}^a(z) \phi_c(z) \nonumber \\
\{Q, \bar{\psi}_a(z)\} &= -i [ \bar{\lambda}\lambda(z) \bar{\psi}_a(z) + \bar{\psi}_a(z) \lambda\bar{\lambda}(z) ] + \epsilon_{bc} \bar{\phi}^b(z) \phi_a(z) \bar{\phi}^c(z). \nonumber
\end{align}
%where we need to remember the constraint
%\begin{align}
%\partial_z \lambda \bar{\lambda}(z) = i [ \epsilon^{ab} \phi_a(z) \bar{\psi}_b(z) - \epsilon_{ab} \psi^a(z) \bar{\phi}^b(z) ] %\equiv \phi \cdot \bar{\psi}(z) - \psi \cdot \bar{\phi}(z) \nonumber \\
%, \quad
%\partial_z \bar{\lambda} \lambda(z) = i [ \epsilon_{ab} \bar{\phi}^a(z) \psi^b(z) - \epsilon^{ab} \bar{\psi}_a(z) \phi_b(z) ]. %\equiv \bar{\phi} \cdot \psi(z) - \bar{\psi} \cdot \phi(z).
%\end{align}
%Whether or not we use \eqref{bilinears1}, there is not an obvious way to assemble the generating functions \eqref{unilinears} into a superfield. In the full ABJM theory, \cite{blm10} had to make some of the R-symmetry non-manifest in order to construct a superfield. Things become nicer if we content ourselves with describing only a subset of the states built out of BPS letters. We can see this by using bilinears everywhere which makes the problem more similar to $\mathcal{N} = 4$ SYM.

The importance of $\lambda \bar{\lambda}(z)$ and $\bar{\lambda} \lambda(z)$ in these transformation laws suggests that our superfield should have one of these (chosen here to be $\lambda \bar{\lambda}(z)$ without loss of generality) as its leading component. It is then necessary to ensure that the rest of the components are bilinears as well. From the form of \eqref{unilinears}, it is straightforward to verify that the generating function for a product is the product of the generating functions. It is therefore convenient to write
\begin{align}
\{Q, \lambda \bar{\lambda}(z)\} &= -i \lambda \bar{\lambda}(z)^2 \label{q-actions3} \\
\{Q, \phi_a \bar{\psi}_b(z)\} &= -i \{\lambda\bar{\lambda}(z), \phi_a \bar{\psi}_b(z)\} + \epsilon_{cd} \phi_a(z) \bar{\phi}^c(z) \phi_b(z) \bar{\phi}^d(z) \nonumber \\
\{Q, \psi^a \bar{\phi}^b(z)\} &= -i \{\lambda\bar{\lambda}(z), \psi^a \bar{\phi}^b(z)\} + \epsilon^{cd} \phi_c(z) \bar{\phi}^a(z) \phi_d(z) \bar{\phi}^b(z) \nonumber \\
[Q, \psi^a \bar{\psi}_b(z)] &= -i [\lambda\bar{\lambda}(z), \psi^a \bar{\psi}_b(z)] + \epsilon^{cd} \phi_c(z) \bar{\phi}^a(z) \phi_d(z) \bar{\psi}_b(z) - \epsilon_{cd} \psi^a(z) \bar{\phi}^c(z) \phi_b(z) \bar{\phi}^d(z) \nonumber \nonumber \\
[Q, \phi_a \bar{\phi}^b(z)] &= -i [\lambda\bar{\lambda}(z), \phi_a \bar{\phi}^b(z)] \nonumber
\end{align}
which follows from \eqref{q-actions1}. Clearly, the relevant actions under $Q$ are now quadratic instead of cubic. We therefore seek a superfield which satisfies
\begin{equation}\label{desired-q-action}
    \{Q,\Psi\bar\Psi(z, \theta)\} = \Psi\bar\Psi(z, \theta)^2.
\end{equation}
Since this is also the $Q$ action in ${\cal N}=4$ SYM, the higher components should be generating functions defined in complete analogy with \eqref{bmn-identifications}.\footnote{With a single Grassmann variable, it is possible to use $\Psi \bar{\Psi}(z, \theta) = -i [\lambda \bar{\lambda}(z) + \theta \phi_a(z) \phi^a(z)]$ as a toy model realizing \eqref{desired-q-action}. We would like to do better since this superfield describes too little of the Hilbert space. It is also possible to show that \eqref{desired-q-action} cannot be satisfied with two Grassmann variables. This agrees with the intuition that there are no ``special'' $2 \times 2$ matrices which come in sets of two.}
\begin{comment}
Mirrorring the ${\cal N}=4$ SYM case, we want to assemble a superfield $\Psi\bar\Psi({\cal Z})$ out of these billinears and have the property
\begin{equation}\label{desired-q-action}
    \{Q,\Psi\bar\Psi({\cal Z})\} = \Psi\bar\Psi({\cal Z})^2.
\end{equation}
This is easy to achieve for a superfield with two degrees of freedom
\begin{align}
& \Psi \bar{\Psi}(z, \theta) = -i \left [ \lambda\bar{\lambda}(z) + \theta \phi \cdot \bar{\phi}(z) \right ] \label{superfield2} \\
& \{Q, \Psi \bar{\Psi}(z, \theta)\} = \Psi \bar{\Psi}(z, \theta)^2, \quad \Psi \bar{\Psi}(0, 0) = 0. \nonumber
\end{align}
but this superfield describes too little of the Hilbert space. One can show that with two degrees os freedom there is no superfield that can be built leading to the quadratic $Q$ action \eqref{desired-q-action}. Nevertheless, with 8 degrees of freedom it is indeed possible.
\end{comment}
Hence, the natural ansatz to propose is
%\begin{align}
%\Psi \bar{\Psi}(z, \theta) = -i[\lambda(z) + 2\theta_n \phi^n(z) + \epsilon^{mnp} \theta_m \theta_n \psi_p(z) + 4\theta_1 \theta_2 \theta_3 f(z)], \label{4d-superfield2}
%\end{align}
\begin{align}
\Psi \bar{\Psi}(z, \theta) =& -i \lambda \bar{\lambda}(z) - i \theta_n \phi_a \bar{\phi}^b(z) (\sigma^n)_b^{\;\;a} + \frac{i}{2} \epsilon^{mnp} \theta_m \theta_n [\phi_a \bar{\psi}_b(z) (\epsilon \sigma_p)^{ab} + \psi^a \bar{\phi}^b(z) (\sigma_p \epsilon)_{ab}] \nonumber \\
&+ \theta_1 \theta_2 \theta_3 [2\psi^a \bar{\psi}_a(z) - \phi_a D \bar{\phi}^a(z) + D\phi_a \bar{\phi}^a(z)]. \label{3d-superfield}
\end{align}
Again, $\Psi \bar{\Psi}(0, 0) = 0$ is the only constraint which needs to be imposed.\footnote{Even though $\epsilon^{ab} \phi_a \bar{\psi}_b(z) - \epsilon_{ab} \psi^a \bar{\phi}^b(z) = -i \lambda \bar{\lambda}(z)$ and $D \phi_a \bar{\phi}^b(z) + \phi_a D \bar{\phi}^b(z) = \partial_z \phi_a \bar{\phi}^b(z)$ must hold, these equations do not introduce additional constraints because not all of their terms appear in \eqref{3d-superfield}.}
To see that squaring this superfield results in the right $Q$ action, we note that the $O(\theta)$, $O(\theta^2)$ and $O(\theta^3)$ components are each transformed into a sum of two terms. The first term consists of (anti-)commutators of the original generating function with $\lambda\bar{\lambda}(z)$ --- this is indeed part of the $Q$ action. As can be seen in \eqref{q-actions3}, it is the part which comes from $Q$ hitting covariant derivatives. The second term is the one from \eqref{n4-qactions} with individual letters promoted to generating functions. The proof that this provides the rest of the $Q$ action (the part where $Q$ does not hit covariant derivatives) follows by repeating \eqref{check-part2} and \eqref{check-part3} \textit{mutatis mutandis}.
%where $b(z)$ is a boson of dimension two that has to be determined by imposing \eqref{desired-q-action}. As argued in Appendix \ref{}, this procedure determines $b(z)$ to be
%\begin{align}
%b(z) = 8 \psi \cdot \bar{\psi}(z) -4i \phi \cdot D \bar{\phi}(z) + 4i D \phi \cdot \bar{\phi}(z).
%\end{align}

In summary, we have shown that a sector of $\frac{1}{12}$-BPS operators in ABJM is described by a $(1|3)$ superfield which obeys the same equations as the $(2|3)$ superfield describing all $\frac{1}{16}$-BPS operators in $\mathcal{N} = 4$ SYM. The mapping between generating functions is
\begin{align}\label{generating-functions-mapping-lift}
\lambda(z) &\mapsto \lambda \bar{\lambda}(z) \\
\phi^n(z) &\mapsto \frac{1}{2} \phi_a \bar{\phi}^b(z) (\sigma^n)_b^{\;\;a} \nonumber \\
\psi_n(z) &\mapsto -\frac{1}{2} \left [ \phi_a \bar{\psi}_b(z) (\epsilon \sigma_n)^{ab} + \psi^a \bar{\phi}^b(z) (\sigma_n \epsilon)_{ab} \right ] \nonumber \\
f(z) &\mapsto \frac{i}{4} \left [ 2 \psi^a \bar{\psi}_a(z) - \phi_a D \bar{\phi}^a(z) + D \phi_a \bar{\phi}^a(z) \right ] \nonumber
\end{align}
where we remind the reader of the definitions \eqref{unilinears}.
\begin{comment}
\begin{align}
\phi^i \equiv \Phi^{4i}, \quad \psi_i \equiv -i \Psi_{i+}, \quad \lambda_{\dot{\alpha}} \equiv \bar{\Psi}^4_{\dot{\alpha}}, \quad f = -iF_{++} \\
\phi_a \equiv 2 \phi_a^L, \quad \bar{\psi}_a \equiv 2 \bar{\psi}^L_{a +}, \quad \bar{\phi}^a \equiv 2 \bar{\phi}^a_R, \quad \psi^a \equiv 2 \psi^a_{R +}. \nonumber
\end{align}
This correspondence might not seem like much since $z$ has one component in the 3d case and two components in the 4d case. However, it is worth remembering that $z$ is there to keep track of covariant derivatives which many known quantum black hole operators in $\mathcal{N} = 4$ SYM do not have. This goes for the $\Delta = \frac{19}{2}$ cohomology class found in \cite{cl22} where an explicit representative was later written in \cite{bmv23}. For more operators with this property, see \cite{cckll23}. Since these results come from setting $z = 0$, we should be able to do the same thing in ABJM theory and find some quantum black holes for free.
\end{comment}
It will be interesting to see if this superfield allows us go beyond \eqref{sym-tower} and lift more fortuitous BPS states in ${\rm U}(N)$ ${\cal N}=4$ SYM to ${\rm U}(N)_k \times {\rm U}(N)_{-k}$ ABJM.
%In particular, the $Q$ action on the ABJM BPS letters on the right-hand side of \eqref{generating-functions-mapping-lift} reproduces the ${\cal N}=4$ SYM $Q$ action on the left-hand side.
It is again helpful to note that there is a fortuitous lift \eqref{subtract-trace} which goes from ${\rm SU}(N)$ to ${\rm U}(N)$ ${\cal N}=4$ SYM.

\section{Conclusion}\label{sec:conclusion}

In this paper, we have investigated fortuitous operators in ABJM theory at large Chern-Simons level using two complementary approaches. First, by adapting the brute-force enumeration algorithm developed in ${\cal N}=4$ Super Yang-Mills theory \cite{2209.06728}, we identified the BPS cohomology classes in \eqref{zfort-n2} for $N = 2$, \eqref{zfort-n3} for $N = 3$ and \eqref{zfort-n4} for $N = 4$. Altogether, this amounts to 244 low-lying fortuitous operators in ABJM. Second, by isolating adjoint-valued bilinear combinations of ABJM BPS letters whose supersymmetry transformations coincide with those used in \cite{0803.4183,1305.6314}, we established a precise map between the non-derivative sector of ${\cal N}=4$ SYM and a somewhat non-standard sector of ABJM theory. This correspondence allows the BMN-type fortuitous operators studied in \cite{2304.10155,2312.16443,2412.08695,gklm25} to be lifted directly to a counterpart in ABJM. For $N = 2$,
-%this yields an infinite tower of fortuitous cohomology classes at high levels from \eqref{sym-tower}, complementing the many that we found at low levels.
-one can complement the low-level search using the operators \eqref{sym-tower} at arbitrarily high levels --- we currently have no reason to believe that the resulting cohomology classes are $Q$-exact, which is the only way for them to not be fortuitous.

If the conjecture of \cite{cl24} relating fortuity to black holes is correct, the Bekenstein-Hawking entropy requires there to be exponentially more fortuitous operators than monotone operators at large $N$. When $N$ is taken to be small, however, there is no guarantee that this hierarchy will still be visible in the low-lying Hilbert space. As such, the abundance of fortuitous operators found here and in \cite{bsvz25} is an obvious advantage to working with the prototypical holographic CFT in 3d instead of 4d. At the same time, ABJM shares a pleasing property with ${\cal N}=4$ SYM --- it admits a weakly-coupled bulk dual which includes propagating gravitons. Together, these facts suggest that ABJM theory will provide a highly desirable setting for studying fortuity in the future. Since the BPS partition functions we quoted were found on a single laptop, it is likely that they can be computed to much higher orders. This is certainly true for $N = 2$, as preliminary testing has shown, and we suspect it will be possible to see statistical evidence for black holes without prohibitively large values of $N$.

For these searches to proceed in earnest, it will be important to construct the actual BPS representatives of the fortuitous cohomology classes. Indeed, it is well-known that there is exactly one operator in each $Q$-cohomology class that is BPS, namely one that is annihilated by $\Delta = \{Q,Q^\dagger\}$. Although our analysis was able to obtain \eqref{abjm-fort2}, \eqref{abjm-fort-singlet} and \eqref{abjm-fort-triplet}, it was not enough to single out the BPS representative. This requires one to diagonalize the two-loop Hamiltonian \eqref{ham-2loop} and look at the zero eigenspace. For the analogous one-loop Hamiltonian of ${\cal N = 4}$ SYM, this was undertaken in the simultaneous works \cite{bmv23,2306.04673}. In \cite{2306.04673}, the Hamiltonian was computed at fixed small values of $N$ by making use of carefully normalized inner products of BPS letters. The distribution of its eigenvalues showed evidence of a gap between classically BPS and one-loop BPS operators, reminiscent of the $O(N^{-2})$ gap obtained in \cite{2203.01331} using the gravitational path integral. For our purposes, it is interesting to note that \cite{2412.03697} has recently found an $O(N^{-3/2})$ gap by studying near-extremal black holes in the gauged supergravity background dual to ABJM theory at $k = 1$. In \cite{bmv23}, the focus was instead on following eigenstates to non-integer values of $N$ where the theory becomes non-unitary and operator dimensions are allowed to acquire imaginary parts. It was found that the fortuitous state corresponding to \eqref{sym-fort} remains in a unitary representation for $N > 2$ but becomes unprotected by ``eating'' a monotone state. Given the high dimensionalities involved, it was beneficial to express the Hamiltonian using the remarkable form developed in \cite{hep-th/0407277}. Similar technology exists for ABJM theory \cite{0901.0411} and we also believe that analytically continuing BPS states with the small charges seen here will be doable without resorting to integrability methods.

It would also be interesting to find truncations of ABJM theory which allow one to identify fortuitous operators without relying on results from other theories. An example would be a sector where the multi-graviton index can be computed exactly at finite $N$. Unfortunately, there are hardly any of these in ${\cal N} = 4$ SYM.
% Say that the Coulomb branch cohomology is now know not to be the full monotone cohomology?
The Schur sector, known for its role in the SCFT/VOA correspondence \cite{1312.5344}, is one example but it is believed to contain no fortuitous operators at all \cite{2310.20086}. So far, it is only the BMN or non-derivative sector which has enabled this type of progress \cite{2304.10155}.\footnote{After making this truncation, \cite{gklm25} showed that there is further ground to be gained by restricting to R-symmetry singlets. This results in a simple monotone index for $N = 2,3,4$ while still leaving many fortuitous states intact.} Even then, the correspondence between the BMN matrix model \cite{hep-th/0202021} and a subsector of ${\cal N} = 4$ SYM is broken by quantum corrections. Since these are now known to affect the $Q$-cohomology \cite{cl25a,cl25b,bk25}, the question of whether the operators in \eqref{sym-tower} are really $\frac{1}{16}$-BPS deserves additional scrutiny. In ABJM theory, the fate of naively $\frac{1}{12}$-BPS operators beyond one loop remains unexplored and we consider this an important direction for future work.

\section*{Acknowledgements}
We are grateful to Pinaki Banerjee, Marco Bianchi, Kasia Budzik, Eduardo Casali, Qi Chen, Matheus Curado, Davide Gaiotto, Zhongjie Huang, Shota Komatsu, Adrian Lopez-Raven, Gabriel Menezes, Harish Murali, Matthew Roberts and Pedro Vieira for stimulating discussions. Research at Perimeter Institute is supported in part by the Government of Canada through the Department of Innovation, Science and Economic Development and by the Province of Ontario through the Ministry of Colleges and Universities. L. G. acknowledges FAPESP Foundation through the grant 2023/04415-2 for support.

\appendix

\section{Conventions for the Superconformal Algebra in $d=3$}\label{app:superconformal-algebra}

In this appendix we establish our conventions for the superconformal algebra in $d=3$, which are the same as the ones in \cite{bcjp20}. Firstly we have the $d=3$ conformal algebra
% B.1
\begin{equation}
    \begin{split}
    [M_{\alpha}^{\phantom{\alpha}\beta},M_{\gamma}^{\phantom{\gamma}\delta}] &= -\delta_{\alpha}^{\phantom{\alpha}\delta}M_{\gamma}^{\phantom{\gamma}\beta} + \delta_{\gamma}^{\phantom{\gamma}\beta}M_{\alpha}^{\phantom{\alpha}\delta},\\
    [M_{\alpha}^{\phantom\alpha\beta},P_{\gamma\delta}] &= \delta_\gamma^{\phantom\gamma \beta}P_{\alpha\delta} + \delta_\delta^{\phantom\delta\beta} P_{\alpha\gamma}-\delta_\alpha^{\phantom{\alpha}\beta} P_{\gamma\delta},\\
    [M_\alpha^{\phantom\alpha\beta},K^{\gamma\delta}] &= - \delta_\alpha^{\phantom\alpha\gamma}K^{\beta\delta} -\delta_{\alpha}^{\phantom\alpha\delta}K ^{\beta\gamma} + \delta_\alpha^{\phantom\alpha\beta}K^{\gamma\delta},\\
    [K^{\alpha\beta},P_{\gamma\delta}] &= 4\delta^{(\alpha}_{(\gamma}M^{\beta)}_{\delta)} + 4\delta^\alpha_{(\gamma}\delta^\beta_{\delta)}D.
    \end{split}
\end{equation}
The associated adjoint relations between the generators are given by
\begin{equation}
    (P_{\alpha\beta})^\dagger = K^{\alpha\beta}, \quad (K^{\alpha\beta})^\dagger = P_{\alpha\beta}, \quad (M_\alpha^{\phantom\alpha\beta})^\dagger = M_{\beta}^{\phantom\beta\alpha}, \quad D^\dagger = D.
\end{equation}
The R-symmetry generators obey an $\mathfrak{so}({\cal N})$ algebra
\begin{equation}
    \begin{split}
        [R_{rs}, R_{tu}] 
&= i \big( 
\delta_{st} R_{ru} 
- \delta_{su} R_{rt} 
- \delta_{rt} R_{su} 
+ \delta_{ru} R_{st} 
\big).
    \end{split}
\end{equation}
The commutators involving bosonic and fermionic generators are
%$Q_{\alpha r}$ and the $\mathfrak{so}(2,3)\oplus \mathfrak{so}({\cal N})$ generators are
\begin{equation}
    \begin{aligned}
        [M_{\alpha}^{\phantom{\alpha}\beta},Q_{\gamma r}] &= \delta_\gamma^{\phantom\gamma\beta}Q_{\alpha r} - \tfrac{1}{2}\delta_\alpha^{\phantom\alpha\beta}Q_{\gamma r}, &\qquad [M_{\alpha}^{\phantom\alpha\beta},S^\gamma_{\phantom\gamma r}] &= -\delta_\alpha^{\phantom\alpha\gamma}S^\beta_{\phantom\beta r}+\tfrac{1}{2}\delta_\alpha^{\phantom\alpha\beta}S^\gamma_{\phantom\gamma r},\\
        [K^{\alpha\beta},Q_{\gamma r}] &=-i(\delta_\gamma^{\phantom\gamma\alpha}S^\beta_{\phantom\beta r}+\delta_\gamma^{\phantom\gamma\beta}S^\alpha_{\phantom\alpha r}), &\qquad [P_{\alpha\beta},S^{\gamma}_{\phantom\gamma r}] &= -i(\delta_\alpha^{\phantom\alpha\gamma}Q_{\beta r}+\delta_\beta^{\phantom\beta\gamma}Q_{\alpha r}),\\
        [D,Q_{\alpha r}] &= \tfrac{1}{2}Q_{\alpha r}, &\qquad [D,S^\alpha_{\phantom\alpha r}] &= - \tfrac{1}{2}S^\alpha_{\phantom\alpha r},\\
        [R_{rs},Q_{\alpha t}] &= i (\delta_{rt}Q_{\alpha s}-\delta_{st}Q_{\alpha r}), &\qquad [R_{rs},S^\alpha_{\phantom\alpha t}] & = i(\delta_{rt}S^\alpha_{\phantom\alpha s}-\delta_{st}S^\alpha_{\phantom\alpha r}),
    \end{aligned}
\end{equation}
\begin{comment}
while the ones involving the $S^\alpha_{\phantom\alpha r}$ and the $\mathfrak{so}(2,3)\oplus \mathfrak{so}({\cal N})$ generators are
\begin{equation}
    \begin{split}
        [M_{\alpha}^{\phantom\alpha\beta},S^\gamma_{\phantom\gamma r}] &= -\delta_\alpha^{\phantom\alpha\gamma}S^\beta_{\phantom\beta r}+\tfrac{1}{2}\delta_\alpha^{\phantom\alpha\beta}S^\gamma_{\phantom\gamma r},\\       [P_{\alpha\beta},S^{\gamma}_{\phantom\gamma r}] &= -i(\delta_\alpha^{\phantom\alpha\gamma}Q_{\beta r}+\delta_\beta^{\phantom\beta\gamma}Q_{\alpha r}),\\
[D,S^\alpha_{\phantom\alpha r}] &= - \tfrac{1}{2}S^\alpha_{\phantom\alpha r},\\
[R_{rs},S^\alpha_{\phantom\alpha t}] & = i(\delta_{rt}S^\alpha_{\phantom\alpha s}-\delta_{st}S^\alpha_{\phantom\alpha r})
    \end{split}
\end{equation}
\end{comment}
while the anti-commutators among the $Q_{\alpha r}$ and $S^\beta_{\phantom\beta s}$ are
\begin{equation}
    \begin{split}
        \{Q_{\alpha r},Q_{\beta s}\} &= 2\delta_{rs}P_{\alpha \beta},\\
\{S^\alpha_{\phantom\alpha r}, S^\beta_{\phantom\beta s}\} &= -2\delta_{rs}K^{\alpha\beta},\\
\{Q_{\alpha r},S^\beta_{\phantom\beta s}\} &= 2i\left[\delta_{rs}(M_\alpha^{\phantom\alpha\beta}+\delta_\alpha^{\phantom\alpha\beta}D)-i\delta_\alpha^{\phantom\alpha\beta}R_{rs}\right].
    \end{split}
\end{equation}
Finally, the BPZ adjoint relations for the supercharges and R-symmetry generators are
\begin{equation}
        (Q_{\alpha r})^\dagger = -i S^\alpha_{\phantom\alpha r}, \quad (S^\alpha_{\phantom\alpha r})^\dagger = -i Q_{\alpha r}, \quad (R_{rs})^\dagger = R_{rs}.
\end{equation}

\section{Centralizer Calculation}\label{app:centralizer-calculation}

In this appendix, we explain the calculation of the centralizer ${\cal C}(\{Q,Q^\dagger\})$ which comes from studying (anti-)commutators. As for bosonic generators, we can see that $R_{12}$ generates a $\mathfrak{u}(1)$, whereas the remaining $R_{ij}$ generate an $\mathfrak{so}(4)$. Using the identifications
\begin{equation}\label{eq:sl2-u1-gen-centralizer}
    {\cal H} = D - M_-^{\phantom - - },\quad P = \frac{1}{2}P_{++},\quad K = -\frac{1}{2}K^{++},\quad C = D+M_{-}^{\phantom - -},
\end{equation}
the algebra
\begin{equation}
    \begin{aligned}
        [M_-^{\phantom --},P_{++}] &=-P_{++}, &\qquad [M_-^{\phantom --},K^{++}] &=K^{++},\\
        [D,P_{++}] &= P_{++}, &\qquad [D,K^{++}] &= - K^{++},\\
        [P_{++},K^{++}] &=4M_-^{\phantom--}-4D,
    \end{aligned}
\end{equation}
satisfied by the conformal generators, becomes
\begin{equation}
    \begin{split}
        [{\cal H},P] &= 2P, \quad [{\cal H},K] = -2K, \quad [P,K] = {\cal H}, \\
        [C,{\cal H}] &= [C,P] = [C,K] = 0.
    \end{split}
\end{equation}
Therefore the bosonic part of ${\cal C}(\{Q,Q^\dagger\})$ is
\begin{equation}
    \mathfrak{sl}(2)\oplus \mathfrak{u}(1)\oplus \mathfrak{u}(1)\oplus \mathfrak{so}(4).
\end{equation}

The fermionic part of ${\cal C}(\{Q,Q^\dagger\})$ is generated by $Q, Q^\dagger$, $Q_{+r}$ and $S^{+}_{\phantom + r}$ with $r\geq 3$. These satisfy the commutation relations
\begin{equation}
    \begin{aligned}
        [M_{-}^{\phantom{-}-},Q_{+ r}] &= - \tfrac{1}{2}Q_{+ r}, &\qquad [M_{-}^{\phantom--},S^+_{\phantom+ r}] &=\tfrac{1}{2}S^+_{\phantom+ r},\\
        [K^{++},Q_{+ r}] &=-2iS^+_{\phantom+r}, &\qquad [P_{++},S^{+}_{\phantom+ r}] &= -2iQ_{+r},\\
        [D,Q_{+ r}] &= \tfrac{1}{2}Q_{+ r}, &\qquad [D,S^+_{\phantom\alpha r}] &= - \tfrac{1}{2}S^+_{\phantom\alpha r},
    \end{aligned}
\end{equation}
and
\begin{equation}
    \begin{aligned}
        [R_{rs},Q_{\alpha t}] &= i (\delta_{rt}Q_{\alpha s}-\delta_{st}Q_{\alpha r}), &\qquad [R_{rs},S^+_{\phantom+ t}] & = i(\delta_{rt}S^+_{\phantom+ s}-\delta_{st}S^+_{\phantom+ r}),\\
	[R_{12},Q_{+ r}] &= 0, &\qquad [R_{12},S^+_{\phantom+ r}] &=0,
    \end{aligned}\label{commutator-set}
\end{equation}
along with the anti-commutation relations
\begin{equation}
    \begin{aligned}
        \{Q_{+ r},Q_{+ s}\} &= 2\delta_{rs}P_{++}, \qquad \{S^+_{\phantom+ r}, S^+_{\phantom+ s}\} = -2\delta_{rs}K^{++},\\
	\{Q_{+ r},S^+_{\phantom+ s}\} &= 2i\left[\delta_{rs}(D + M_+^{\phantom++})-iR_{rs}\right] \\
	&= 2i\left[\delta_{rs}(D - M_-^{\phantom--})-iR_{rs}\right].
    \end{aligned}\label{anti-commutator-set}
\end{equation}
We have recognized that the last line of \eqref{commutator-set} is zero because $r \geq 3$ in $Q_{+r}$ and $S^+_{\phantom+ r}$. We have also used $M_{+}^{\phantom++} = -M_-^{\phantom - -}$ in \eqref{anti-commutator-set}.
Recasting all commutators in terms of the $\mathfrak{sl}(2)$ generators defined before on \eqref{eq:sl2-u1-gen-centralizer} leads to
\begin{equation}
    \begin{aligned}
        [{\cal H},Q_{+ r}] &= Q_{+ r}, &\qquad [{\cal H},S^+_{\phantom+ r}] &=-S^+_{\phantom+ r},\\
\{Q_{+ r},Q_{+ s}\} &= 4\delta_{rs}P, &\qquad \{S^+_{\phantom+ r}, S^+_{\phantom+ s}\} &= 4\delta_{rs}K,\\
[R_{rs},Q_{\alpha t}] &= i (\delta_{rt}Q_{\alpha s}-\delta_{st}Q_{\alpha r}), &\qquad [R_{rs},S^+_{\phantom+ t}] & = i(\delta_{rt}S^+_{\phantom+ s}-\delta_{st}S^+_{\phantom+ r}),\\
\{Q_{+ r},S^+_{\phantom+ s}\} &= 2i\left(\delta_{rs}{\cal H}-iR_{rs}\right),
    \end{aligned}
\end{equation}
where we only listed the non-trivial ones, noting that $C$ and $R_{12}$ commute with all the $Q_{+r}$ and $S^+_{\phantom+ r}$.

Now we finally turn to $Q$ and $Q^\dagger$. The commutators
\begin{equation}
    \begin{aligned}
        [M_{-}^{\phantom{-}-},Q_{- r}] &= \tfrac{1}{2}Q_{- r}, &\qquad [M_{-}^{\phantom--},S^-_{\phantom- r}] &= -\tfrac{1}{2}S^-_{\phantom{-}r},\\
        [K^{++},Q_{-r}] &=0, &\qquad [P_{++},S^{-}_{\phantom- r}] &= 0,\\
        [D,Q_{- r}] &= \tfrac{1}{2}Q_{- r}, &\qquad [D,S^-_{\phantom- r}] &= - \tfrac{1}{2}S^-_{\phantom- r}
    \end{aligned}
\end{equation}
for $r\in \{1,2\}$ become
\begin{equation}
    \begin{split}
        & [{\cal H},Q] = [P,Q] = [K,Q] = 0, \quad [C,Q] = Q,\\
	& [{\cal H},Q^\dagger] = [P,Q^\dagger] = [K,Q^\dagger] = 0, \quad [C,Q^\dagger] = -Q^\dagger
    \end{split}
\end{equation}
after using \eqref{eq:sl2-u1-gen-centralizer} and the fact that $Q, Q^\dagger = Q_{-1} \pm i Q_{-2}$. When it comes to R-symmetry, we can take $r,s\notin \{1,2\}$ and $t\in \{1,2\}$ to write
\begin{equation}
    \begin{aligned}
                [R_{12},Q_{-t}] &= i (\delta_{1t}Q_{- 2}-\delta_{2t}Q_{-1}), &\qquad [R_{12},S^-_{\phantom- t}] & = i(\delta_{1t}S^-_{\phantom- 2}-\delta_{2t}S^-_{\phantom- 1}),\\
		[R_{rs},Q_{-t}] &= 0, &\qquad [R_{rs},S^-_{\phantom- t}] &= 0.
    \end{aligned}
\end{equation}
These become
\begin{equation}
\begin{split}
	& [R_{rs},Q] = [R_{rs},Q^\dagger] = 0,\quad r,s\geq 3, \\
	& [R_{12},Q] = Q, \quad [R_{12},Q^\dagger] = -Q^\dagger
\end{split}
\end{equation}
when we again use \eqref{eq:sl2-u1-gen-centralizer} and $Q, Q^\dagger = Q_{-1} \pm i Q_{-2}$.

It is also easy to see that all commutators mixing $Q$, $Q^\dagger$ with $Q_{+r}$ and $S^{-}_{\phantom-r}$ are all zero. The last non-trivial commutator is exactly $\{Q,Q^\dagger\}$, which in terms of $C$ and $R_{12}$ is
\begin{equation}
    \{Q,Q^\dagger\} = 4C - 4R_{12}
\end{equation}
We have thus found that the bosonic part of ${\cal C}(\{Q,Q^\dagger\})$ is $\mathfrak{sl}(2)\oplus \mathfrak{u}(1)_C\oplus\mathfrak{u}(1)_R\oplus\mathfrak{so}(4)_R$.
%where the notation reminds us that $\mathfrak{u}(1)_C$ is the $\mathfrak{u}(1)$ that comes from the $C$ generator inherited from the conformal algebra $\mathfrak{so}(2,3)$, whereas $\mathfrak{u}(1)_R$ is the $\mathfrak{u}(1)$ that comes from the $R$-symmetry.
The fermionic generators $Q_{+r}$ and $S^+_{\phantom+r}$ with $r\geq 3$ have zero $\mathfrak{u}(1)_C$ charge and, together with the $\mathfrak{sl}(2)$ generators, they form a one-dimensional superconformal algebra with ${\cal N}=4$. So we group $\{{\cal H},P,K,Q_{+r},S^+_{\phantom+r},R_{rs} : r,s\geq 3\}$ together in this superconformal algebra $\mathfrak{osp}(4|2)$. The remaining generators are $\{Q,Q^\dagger,C, R_{12}\}$ and together they form a $\mathfrak{u}(1|1)$ algebra. Altogether,
\begin{equation}
    {\cal C}(\{Q,Q^\dagger\})\simeq \mathfrak{osp}(4|2)\oplus \mathfrak{u}(1|1).
\end{equation}

\section{The ``BPS - Cohomology'' Isomorphism}\label{app:bps-cohomology-isomorphism}

In this appendix, we are going to prove one important fact, namely that ${\cal H}_N^{\rm BPS}$ is isomorphic to a direct sum of several cohomology spaces associated to one specific co-chain complex. In what follows, all that we assume is a Hilbert space ${\cal H}$ with a nilpotent operator $Q\in \operatorname{End}({\cal H})$ and $\Delta = \{Q,Q^\dagger\}$.\footnote{Note that in this appendix we have rescaled $\Delta$ by an irrelevant factor of $2$ compared with the definition used in the main text.} Note that for a fully rigorous mathematical proof one would have to be careful about the domains of the various operators involved, as it is common in functional analysis. We will sidestep this issue, thereby providing a proof at the physicist level of rigor, manipulating everything as in traditional quantum mechanics.

First, we show some elementary properties of the operator $\Delta$:
\begin{prop}
    The operator $\Delta = \{Q,Q^\dagger\}$ is self-adjoint, positive and commutes with $Q$ and $Q^\dagger$.
\end{prop}
\begin{proof}
    The property $\Delta^\dagger = \Delta$ is manifest.
\begin{comment}
    The adjoint can be computed straightforwardly
    \begin{equation}
        \begin{split}
            \Delta^\dagger &= (QQ^\dagger + Q^\dagger Q)^\dagger\\
            &= (Q^\dagger)^\dagger Q^\dagger + Q^\dagger (Q^\dagger)^\dagger\\
            &= Q Q^\dagger + Q^\dagger Q\\
            &= \Delta.
        \end{split}
    \end{equation}
\end{comment}
    The commutator also follows from a simple calculation using nilpotency
    \begin{equation}
        \begin{split}
            [\Delta,Q] &= (QQ^\dagger +Q^\dagger Q)Q - Q(QQ^\dagger + Q^\dagger Q) = 0
        \end{split}
    \end{equation}
    with the same calculation showing that $[\Delta,Q^\dagger]=0$. To show positivity take $\psi \in {\cal H}$, then
    \begin{equation}
        \begin{split}
            \langle \psi,\Delta\psi\rangle &= \langle \psi,QQ^\dagger \psi \rangle + \langle \psi,Q^\dagger Q\psi\rangle\\
            &= \langle Q^\dagger \psi,Q^\dagger \psi\rangle + \langle Q\psi,Q\psi\rangle\\
            &= \|Q\psi\|^2 + \|Q^\dagger \psi\|^2\\
            &\geq 0,
        \end{split}
    \end{equation}
    which establishes positivity.
\end{proof}
\begin{prop}\label{prop:ker-delta-ker-q}
    Let $\psi\in \ker\Delta$, then $Q\psi=0$ and $Q^\dagger\psi =0$.
\end{prop}
\begin{proof}
    Since $\Delta\psi =0$, we have $\langle \psi,\Delta\psi\rangle=0$, and hence by the last proposition
    \begin{equation}
        \|Q\psi\|^2 + \|Q^\dagger \psi\|^2 = 0.
    \end{equation}
    As such, we must have both $Q\psi=0$ and $Q^\dagger\psi=0$.
\end{proof}
\begin{prop}\label{prop:hodge-decomposition}
    The subspaces $\operatorname{im}Q$ and $\operatorname{im}Q^\dagger$ are orthogonal in ${\cal H}$. Moreover, if $\Delta = \{Q,Q^\dagger\}$, $\ker \Delta = (\operatorname{im}Q\oplus \operatorname{im}Q^\dagger)^\perp$ implying that the Hilbert space decomposes as ${\cal H} = \operatorname{im}Q\oplus\operatorname{im}Q^\dagger \oplus \ker \Delta$.
\end{prop}
\begin{proof}
    Take arbitrary $\psi,\phi\in {\cal H}$. Then
    \begin{equation}
        \langle Q\psi,Q^\dagger \phi\rangle = \langle Q^2\psi,\phi\rangle = 0
    \end{equation}
    by using nilpotency $Q^2=0$. This shows orthogonality of $\operatorname{im}Q$ and $\operatorname{im}Q^\dagger$. Now define ${\cal H}_0 \equiv \operatorname{im}Q\oplus \operatorname{im}Q^\dagger$. We show that $\ker \Delta \subset {\cal H}_0^\perp$. Take $\psi \in \ker \Delta$ and $\chi \in {\cal H}$, then
    \begin{equation}
        \langle \psi,Q\chi\rangle = \langle Q^\dagger \psi,\chi\rangle = 0
    \end{equation}
    because $\psi\in \ker\Delta$ implies $Q^\dagger\psi=0$. This shows that $\psi\in (\operatorname{im}Q)^\perp$ and the same argument applied to $\langle\psi,Q^\dagger\chi\rangle$ shows that $\psi\in (\operatorname{im}Q^\dagger)^\perp$. As such $\psi \in {\cal H}_0^\perp$ and hence $\ker\Delta \subset {\cal H}_0^\perp$. For the other direction, let $\psi \in {\cal H}_0^\perp$, then
    \begin{equation}
        \langle \psi,Q\chi\rangle = \langle\psi,Q^\dagger \chi\rangle = 0
    \end{equation}
    for all $\chi$. This means $\langle Q\psi,\chi\rangle = 0$ and $\langle Q^\dagger \psi,\chi\rangle=0$ for all $\chi$, and hence $Q\psi=0$ and $Q^\dagger\psi=0$. This implies that  $\psi \in \ker \Delta$ and hence ${\cal H}_0^\perp \subset \ker \Delta$. As a result, we showed $\ker \Delta = (\operatorname{im}Q\oplus\operatorname{im}Q^\dagger)^\perp$.
\end{proof}
\begin{comment}
\begin{prop}\label{prop:hodge-decomposition}
    The Hilbert space decomposes as ${\cal H} = \operatorname{im}Q\oplus\operatorname{im}Q^\dagger \oplus \ker \Delta$
\end{prop}
\begin{proof}
    Let ${\cal H}_0 = \operatorname{im}Q\oplus\operatorname{im}Q^\dagger$. For any subspace it holds ${\cal H} = {\cal H}_0\oplus{\cal H}_0^\perp$ and by the last proposition it turns out we find
    \begin{equation}
        {\cal H} = \operatorname{im}Q \oplus \operatorname{im}Q^\dagger \oplus \ker \Delta.
    \end{equation}
\end{proof}
\end{comment}
Next we are going to show that $\ker \Delta$ can be constructed from the $Q$-cohomology. In general, to define cohomology we need a complex. What we are considering as cohomology at this stage, however, is just the quotient
\begin{equation}
    H_Q({\cal H}) \equiv \dfrac{\ker Q}{\operatorname{im} Q},
\end{equation}
which is the only cohomology space of the trivial complex
\begin{center}
    \begin{tikzcd}
        \cdots \arrow[r] &0 \arrow[r]& {\cal H}\arrow[r, "Q"] & {\cal H} \arrow[r] & 0 \arrow[r]& \cdots
    \end{tikzcd}
\end{center}
The structure of a non-trivial complex appears later under the assumption of extra structure. In this scenario we have the following theorem.
\begin{thm}
    Let ${\cal H}$ be a Hilbert space and $Q\in \operatorname{End}({\cal H})$ a nilpotent operator. Define $\Delta = \{Q,Q^\dagger\}$ and
    \begin{equation}
        H_Q({\cal H}) \equiv \dfrac{\ker Q}{\operatorname{im} Q}.
    \end{equation}
    In these conditions there is an isomorphism $\ker \Delta \simeq H_Q({\cal H})$.
\end{thm}
\begin{proof}
    Define the natural projection map $f: \ker\Delta \to H_Q({\cal H})$ by
    \begin{equation}
        f(\psi) = [\psi].
    \end{equation}
    We show that $f$ is an isomorphism. We first show injectivity. Assume that $f(\psi)=0$. Then $[\psi]=0$ and therefore $\psi = Q\chi$ for some $\chi \in {\cal H}$. But $\psi\in \ker \Delta$ and by Proposition \ref{prop:ker-delta-ker-q} this implies $Q^\dagger\psi=0$. This means $Q^\dagger Q\chi=0$ and therefore
    \begin{equation}
        \langle \chi,Q^\dagger Q\chi\rangle = 0\Longrightarrow \langle Q\chi,Q\chi\rangle =0.
    \end{equation}
    As such $\|Q\chi\|^2 =0$ and therefore $Q\chi=0$, but then $\psi=0$ and $f$ is injective.

    To show surjectivity, let $[\psi]$ be an arbitrary class, we must show that there is a class representative $\chi$ with $\chi\in \ker\Delta$. To do that, use the decomposition in Proposition \ref{prop:hodge-decomposition}:
    \begin{equation}
        \psi = Q\alpha + Q^\dagger \beta + \chi,
    \end{equation}
    where $\chi \in \ker\Delta$. Since $[\psi]\in H_Q({\cal H})$ we must have $Q\psi=0$. As such, acting with $Q$
    \begin{equation}
       \begin{split}
            Q\psi &= Q^2\alpha + QQ^\dagger \beta + Q\chi = 0
       \end{split}
    \end{equation}
    But now $Q^2\alpha=0$ and $Q\chi=0$ because $\chi \in \ker\Delta$ by Proposition \ref{prop:ker-delta-ker-q}. As such $QQ^\dagger\beta=0$, but then
    \begin{equation}
        \langle \beta,QQ^\dagger\beta\rangle =0\Longrightarrow \langle Q^\dagger\beta,Q^\dagger\beta\rangle =0
    \end{equation}
    which implies $\|Q^\dagger\beta\|^2=0$ and hence $Q^\dagger\beta=0$. As a result
    \begin{equation}
        \psi = \chi + Q\alpha
    \end{equation}
    where $\chi \in \ker\Delta$. Since $\psi$ and $\chi$ differ by $Q\alpha$ it follows $[\psi]=[\chi]$ and hence
    \begin{equation}
        f(\chi+Q\alpha)= [\psi],
    \end{equation}
    which shows that $f$ is surjective. we have thus established $f$ as an isomorphism.
\end{proof}
Note that this theorem tells us that there is one and only one BPS representative of each $Q$-cohomology class. Finally we explain how the structure of a non-trivial complex in \cite{cl24} appears. To that end, we let $E\in \operatorname{End}({\cal H})$ be one bounded hermitian operator, with discrete spectrum $\sigma(E) = \{E_n : n\in \mathbb{Z}\}$ that takes the form $E_n = E_0 + n$, and such that $[E,Q] = Q$. We decompose ${\cal H}$ into the eigenspaces of $E$
\begin{equation}
    {\cal H} = \bigoplus_{n\in \mathbb{Z}} {\cal H}^{(n)}
\end{equation}
where $E\psi = E_n \psi$ for $\psi \in {\cal H}^{(n)}$. The condition $[E,Q]=Q$ implies that if $\psi \in {\cal H}^{(n)}$ then $Q\psi\in {\cal H}^{(n+1)}$. This means that ${\cal H}^\bullet = ({\cal H}^{(n)},Q)$ forms a complex with cohomology groups $H_Q^{(n)}({\cal H}^\bullet)$ and it is clear that their direct sum $H_Q^\ast({\cal H}^\bullet)$ is isomorphic to $H_Q({\cal H})$.

 \bibliographystyle{utphys}
 \bibliography{references}

\end{document}